\documentclass[11pt,draftcls,onecolumn]{IEEEtran}
\usepackage{epsfig}
\usepackage{amsmath}
\usepackage{amsmath, amssymb, bbm, xspace}
\usepackage{epsfig}
\usepackage{longtable}
\usepackage{color}
\usepackage{mathrsfs}
\usepackage{comment}
\usepackage{ifthen}
\newboolean{showcomments}
\setboolean{showcomments}{true}
\newcommand{\aak}[1]{  \ifthenelse{\boolean{showcomments}}
{ \textcolor{blue}{(AAK says:  #1)}} {}  }

\usepackage{courier}
\newtheorem{theorem}{Theorem}[section]
\newtheorem{assumption}{Assumption}[section]
\newtheorem{lemma}[theorem]{Lemma}

\newtheorem{definition}{Definition}[section]

\def\bkE{{\rm I\kern-.17em E}}
\def\bk1{{\rm 1\kern-.17em l}}
\def\bkD{{\rm I\kern-.17em D}}
\def\bkR{{\rm I\kern-.17em R}}
\def\bkP{{\rm I\kern-.17em P}}

\def\Cbb{{\mathbb{C}}}

\def\bkZ{{\bf{Z}}}

\def\bfone{{\bf 1}}

\def\bfU{{\bf U}}
\def\bfC{{\bf C}}

\def\bkE{{\rm I\kern-.17em E}}
\def\bk1{{\rm 1\kern-.17em l}}
\def\bkD{{\rm I\kern-.17em D}}
\def\bkR{{\rm I\kern-.17em R}}
\def\bkP{{\rm I\kern-.17em P}}

\makeatletter
\newcommand{\pushright}[1]{\ifmeasuring@#1\else\omit\hfill$\displaystyle#1$\fi\ignorespaces}
\newcommand{\pushleft}[1]{\ifmeasuring@#1\else\omit$\displaystyle#1$\hfill\fi\ignorespaces}
\makeatother

\newcommand{\scg}{shared-constraint game\@\xspace} %%% i.e.,
\newcommand{\scgs}{shared-constraint games\@\xspace}
\newcommand{\Scgs}{Shared-constraint games\@\xspace}
\def\bkZ{{\bf{Z}}}
\def\b12{(\beta_1,\beta_2)}

\def\bfA{{\bf A}}

\newenvironment{example}{{\noindent \bf Example}}{\hfill $\square$\hspace{-4.5pt}\vspace{6pt}}
\newcounter{example}
\renewcommand{\theexample}{\thesection.\arabic{example}}
\newenvironment{examplec}[1][]{\refstepcounter{example}
\par\medskip \noindent%
   \textbf{Example~\theexample. #1} \rmfamily}{\hfill $\square$   \hspace{-4.5pt} \vspace{6pt}}

\newcounter{remark}
\renewcommand{\theremark}{\thesection.\arabic{remark}}

\newlength{\noteWidth}
\long\def\notes#1{\ifinner
{\tiny #1}
\else
\marginpar{\parbox[t]{\noteWidth}{\raggedright\tiny #1}}
\fi\typeout{#1}}

 \def\notes#1{\typeout{read notes: #1}} %uncomment for final version

\def\mi{^{-i}}
\newcommand{\ie}{i.e.\@\xspace} %%% i.e.,
\newcommand{\eg}{e.g.\@\xspace} %%% e.g.,
\newcommand{\etal}{et al.\@\xspace} %%% e.g., Gill \etal (1986)

\newcommand{\Real}{\ensuremath{\mathbb{R}}}

\newcommand{\maximize}[1]{\displaystyle\maxim_{#1}}
\newcommand{\maxim}{\mathop{\hbox{\rm max}}}

\def\subject{\hbox{\rm s.t.}}

\def\Cbb{\mathbb{C}}

\def\Cbb{{\mathbb{C}}}

\def\exp{\mathop{\hbox{\rm exp}}}

\def\spose#1{\hbox to 0pt{#1\hss}}

\def\text #1{\hbox{\quad#1\quad}}

\def\nthinsp{\mskip -2   mu}

\def\superstar{^{\raise 0.5pt\hbox{$\nthinsp *$}}}
\def\SUPERSTAR{^{\raise 0.5pt\hbox{$*$}}}
\def\lamstarT {\lambda^{\raise 0.5pt\hbox{$\nthinsp *$}T}}

\def\Fscr{{\cal F}}

\def\Iscr{{\cal I}}

\def\Lscr{{\cal L}}

\def\Nscr{{\cal N}}

\def\aur{\;\textrm{and}\;}

\def\non{\nonumber}
\def\SOL{{\rm SOL}}
\def\VI{{\rm VI}}

\def\VIs{{\rm VIs}}
\def\QVI{{\rm QVI}}

\def\GNE{{\rm GNE}}
\def\VE{{\rm VE}}

\def\SYS{{\rm (SYS)}}

\let\forallnew\forall
\renewcommand{\forall}{\forallnew\ }
\let\forall\forallnew

\def\ds{\displaystyle}
		\def\bkE{{\rm I\kern-.17em E}}
		\def\bk1{{\rm 1\kern-.17em l}}
		\def\bkD{{\rm I\kern-.17em D}}
		\def\bkR{{\rm I\kern-.17em R}}
		\def\bkP{{\rm I\kern-.17em P}}
		\def\bkY{{\bf \kern-.17em Y}}
		\def\bkZ{{\bf \kern-.17em Z}}
		\def\bkC{{\bf  \kern-.17em C}}
{\begin{list}{}%
         {\setlength{\leftmargin}{#1}}%
         \item[]%
}
{\end{list}}

		\def\bsp{\begin{split}}
		\def\beq{\begin{eqnarray}}
		\def\bal{\begin{align*}}
		\def\bc{\begin{center}}
		\def\be{\begin{enumerate}}
		\def\bi{\begin{itemize}}
		\def\bs{\begin{small}}
		\def\bS{\begin{slide}}
		\def\ec{\end{center}}
		\def\ee{\end{enumerate}}
		\def\ei{\end{itemize}}
		\def\es{\end{small}}
		\def\eS{\end{slide}}
		\def\eeq{\end{eqnarray}}
		\def\eal{\end{align*}}
		\def\esp{\end{split}}
		\def\qed{ \vrule height7.5pt width7.5pt depth0pt}  %width4.17pt depth0pt} 

\newcommand{\boxedeqnnew}[2]{\vspace{12pt}
 \noindent\framebox[\textwidth]{\begin{tabular*}{\textwidth}{p{.8\textwidth}p{.15\textwidth}}
                             \hspace{.05 \textwidth}    #1 & \raggedleft #2
                                \end{tabular*}
}\vspace{12pt}
 }

\newcommand{\boxedeqnsmall}[6]{\vspace{12pt}
 \noindent\framebox[#3 \textwidth]{\begin{tabular*}{#3 \textwidth}{p{#5 \textwidth }p{#6 \textwidth }}
                               \hspace{#4 \textwidth}  #1 & \raggedleft #2
                                \end{tabular*}
}\vspace{12pt}
 }

\def\maxproblem#1#2#3#4{\fbox
		 {\begin{tabular*}{0.80\textwidth}
			{@{}l@{\extracolsep{\fill}}l@{\extracolsep{6pt}}l@{\extracolsep{\fill}}c@{}}
				#1 & $\maximize{#2}$ & $#3$ & $ $ \\[5pt]
					 & $\subject\ $    & $#4$ & $ $
			\end{tabular*}}
			}

	\def\cp2problem#1#2#3#4{\fbox
		 {\begin{tabular*}{0.9\textwidth}
			{@{}l@{\extracolsep{\fill}}l@{\extracolsep{6pt}}l@{\extracolsep{\fill}}c@{}}
				#1 & & $#4 $ 
			\end{tabular*}}}

\newcommand{\pmat}[1]{\begin{pmatrix} #1 \end{pmatrix}}

		\def\bkE{{\rm I\kern-.17em E}}
		\def\bk1{{\rm 1\kern-.17em l}}
		\def\bkD{{\rm I\kern-.17em D}}
		\def\bkR{{\rm I\kern-.17em R}}
		\def\bkP{{\rm I\kern-.17em P}}
		
		\def\bkZ{{\bf{Z}}}

\newcommand {\beeq}[1]{\begin{equation}\label{#1}}
\newcommand {\eeeq}{\end{equation}}
\newcommand {\bea}{\begin{eqnarray}}
\newcommand {\eea}{\end{eqnarray}}

\def\texitem#1{\par\smallskip\noindent\hangindent 25pt
               \hbox to 25pt {\hss #1 ~}\ignorespaces}

\def\bsp{\begin{split}}
		\def\beq{\begin{eqnarray}}
		\def\bal{\begin{align*}}
		\def\bc{\begin{center}}
		\def\be{\begin{enumerate}}
		\def\bi{\begin{itemize}}
		\def\bs{\begin{small}}
		\def\bS{\begin{slide}}
		\def\ec{\end{center}}
		\def\ee{\end{enumerate}}
		\def\ei{\end{itemize}}
		\def\es{\end{small}}
		\def\eS{\end{slide}}
		\def\eeq{\end{eqnarray}}
		\def\eal{\end{align*}}
		\def\esp{\end{split}}
		\def\qed{ \vrule height7.5pt width7.5pt depth0pt}  %width4.17pt depth0pt} 

\usepackage{amsmath, amssymb, xspace}
\usepackage{epsfig}
\usepackage{longtable}
\usepackage{color}
\usepackage{mathrsfs}
\usepackage{subfig}
\newenvironment{proof}[1][]{{\noindent \emph {Proof} #1: }}{\hfill \qed \vspace{3pt}\\ }

\def\Nscr{{\cal N}}
\def\subject{\hbox{\rm subject to}}

\usepackage{epsfig}
\usepackage{longtable}
\usepackage{color}

\begin{document}
\title{The Efficiency of Generalized Nash and Variational Equilibria} 
\author{Ankur A. Kulkarni\thanks{Systems and Control Engineering, Indian Institute of Technology Bombay, Mumbai, India 400076. The author can be contacted at \texttt{\small kulkarni.ankur@iitb.ac.in}, or +91-9167889384 (phone) or +91-2225720057 (fax)}}
\date{}
\maketitle
\begin{abstract}
Shared-constraint games are noncooperative $N$-player games where players are coupled through a common coupling constraint. It is known that such games admit two kinds of equilibria -- generalized Nash equilibria (GNE) and variational equilibria (VE) -- with two different economic interpretations. 
We consider such games in the context of resource allocation, where players move simultaneously to decide portions of the resource they can consume under a coupling constraint that the sum of the portions they demand be no more than the capacity of the resource. We clarify the worst case and best case efficiency of these kinds of equilibria over all games in a class. We find that the worst case efficiency of both solution concepts in zero and the best case efficiency is unity. Moreover, we characterize the subclass of games where all VE are efficient and show that even in this subclass but the worst case efficiency of GNE is zero. We finally discuss means by which zero worst case efficiency can be remedied.
\end{abstract}

\section{Introduction}
Generalized Nash games with shared constraints are noncooperative $N$-player games in which the strategies available 
to each player are constrained by the requirement that the $N$-tuple of player-strategies belong to a common set. 
This work considers shared-constraint games which arise from the competition for a finite resource, 
such as bandwidth or energy. The shared constraint is the common requirement that the portions of the 
resource allocated to all players must total to not exceed the available quantity of the resource. 
We investigate and clarify the efficiency of the equilibria of these games with respect to a socially desirable  
outcome. 

Our game may be thought of as a {\em resource allocation game}, though this term has acquired a somewhat more 
specific meaning~\cite{johari04efficiency} that does not commonly allude to our setting. 
The traditional  resource allocation game has players that compete 
in a noncooperative fashion to access the resource, but their competition is mediated by a {\em mechanism}. 
A mechanism accepts bids from players and allocates the resource to the players in exchange for a payment, thereby also maintaining feasibility of the allocation. Since players may be price anticipating, the choice of a player's bid is strategically coupled 
to the bids of other players, leading to a game in which the strategy of each player is its bid, its constraint is the space of its bids 
and its objective is to maximize its utility less the payment charged by the mechanism. In contrast, our 
game is played in the space of resource portions, the players are cognizant of the finiteness of the resource 
and their objective is to maximize their utility.

 In a resource-allocation game, each player desires to 
maximize its utility. But from the {system-perspective}, the 
goal is to allocate resources in a manner that a prescribed societal objective, characterized by the 
optimization of a {\em welfare function},  is met. A natural and commonly employed
societal objective is the optimzation of {\em aggregate} utility or sum of utilities of all players, \ie the maximization of social welfare, 
though other objectives could be chosen based on the setting. 
When utilities are measured in monetary terms and arbitrary money transfers are permitted across agents, 
aggregate utility maximization also equivalent to the classical concept of {\em Pareto efficiency}\footnote{A Pareto efficient allocation is an allocation, by deviations from which the value derived by any player cannot be strictly increased without simultaneously resulting in strictly 
lower value for another player.
}~\cite{johari04thesis} or {\em economic efficiency}. We say that an allocation is {\em efficient} if it implements the aggregate utility 
and we call the {\em efficiency} of an allocation, the ratio of the aggregate utility for this allocation to that of an efficient allocation.

% The question of allocation of resources is thus constituted by two levels of optimization problems, one at the level of the players 
% and the other at the level of the system. At the level of the players, 
% each player attempts to optimize its utility and at the system-level, there is a social planner's or system 
% administrator's problem of optimizing efficiency. Both problems are constrained. Each player faces 
% the constraint of the admissible bids. The social planner faces the constraint of allocating resources without 
% exceeding capacity. 

The allocation corresponding to the Nash equilibrium of the game is, in 
general, only feasible for the social planner's problem of maximizing aggregate utility. 
Thus these allocations are, in general, 
not efficient. The question of resource allocation is therefore one with dual, and  
possibly conflicting objectives of efficiency and competition. The usual approach to this dilemma gives primacy to competition while 
attempting to control the loss in efficiency. Specifically, when allocating resources through a mechanism, the effort is to 
design a mechanism that provides a guarantee of low efficiency loss at its Nash equilibrium. This is the central pursuit 
of the field of {\em mechanism design}~\cite{narahari09game}, though the societal objective to be met differs according 
to situation.

Our game 
%, though dealing with allocation of resources, 
represents a departure from this approach. 
Though we also have the dual objectives of competition and efficiency, our game is relevant if one makes the assumption that the 
option of a mechanism does not exist. In our setting, players move simultaneously, in a noncooperative manner and 
compete for portions of the resource. A shared-constraint game naturally fits such an 
interaction. The goal behind our work is not to suggest this game as an alternative to mechanism  driven approaches
to resource allocation, but instead to analyze the specific setting where it is relevant and to study 
the efficiency of equilibria that result from it. 
The subtlety though is that shared-constraint games admit two kinds of solution concepts, 
generalized Nash equilibrium (\GNE) and the variational equilibrium (\VE), with 
two different economic interpretations (see \cite{kulkarni09refinement} and~\cite{kulkarni12revisiting}). In~\cite{kulkarni09refinement}, a theory was presented to argue that the \VE can be  considered a refinement of the \GNE for a large class of games. In the light of this, our goal is also to do a
comparative study of 
these classes of equilibria with respect to the metric of efficiency. 
Though our work is not intended as an alternative to other mechanism-based approaches, it contributes to 
identifying settings where the VEs of this \scg have higher efficiency 
than the equilibria induced by mechanisms. 

% The distinction between the GNE and the VE and was discussed 
% in detail in \cite{kulkarni09refinement}.
 We discuss the efficiency of these equilibria separately. 
In particular, we are concerned with the best case efficiency and the worst case efficiency of the GNE 
and the \VE. Ours is, to the best of our knowledge, the only work on the efficiency of equilibria in general 
\scgs, though the social welfare of the equilibria of these games has been considered in the setting of
congestion control~\cite{alpcan02game}. A detailed comparison with other work is encompassed in Section~\ref{sec:relation}. 

We consider two kinds of utility functions. In the first kind, the utility derived by any player is a function only of 
the allocation it receives. In the second, and more general setting, utilities are dependent on allocations received by other players. 
The first setting has an interesting interpretation in terms of the {\em competitive} equilibrium. We also see that 
the VE is efficient in this setting, whereas the GNE can be arbitrarily inefficient.  
In the more general setting we are concerned with the best case efficiency and worst case efficiency over a class of 
utility functions. We characterize settings under which full efficiency is obtained and identify that the solution 
concept of the VE achieves this unit efficiency. We also show that a departure from this setting can lead to arbitrarily 
low efficiency. Specifically, if one considers the GNE as a solution concept, for settings where the VE is efficient, 
one can get arbitrarily low efficiency. And a departure from the ``efficient'' setting of the VE can lead to arbitrarily low 
efficiency for the VE. 

We then suggest some ways in which this arbitrarily low efficiency may be remedied. 
One of the causes of low efficiency is the possibility of wide variation in the marginal utilities of players. 
A more restricted class of utility functions 
% in which the gradient map of every member function is
% bounded away from zero and from above uniformly over the domain, 
gives a more favorable worst case efficiency. 
% The choice of the social choice function  is not 
% universal and notions of efficiency with respect to other system 
% problems may be more relevant in other settings. 
We consider a setting, inspired by \cite{alpcan02game}, with a slightly different social choice function. 
This setting is a game where players incur costs that, from the system point of view, are not additive. 
Thus the social choice function 
is not merely the sum of the objectives of all players. We characterize utility functions
for which the VE is efficient under this notion of efficiency. Finally we consider the imposition of a reserve price on
players. The reserve price has the effect of eliminating players with low interest in the resource.  
The GNE of the resulting game is more indicative of the system optimal solution. We find that under certain conditions, efficiency
as high as unity is obtainable by the imposition of an appropriate reserve price.

% Finally, we study 
% a generalized Nash-Cournot game for resource allocation and compute the efficiency for this game.

In Section~\ref{sec:preliminaries} we introduce the model and our assumptions. 
In Section~\ref{sec:relation} we clarify the relationship of our approach with the competitive equilibrium 
and other approaches that use mechanisms. Following that we give a general bound on the efficiency by showing that 
In Section~\ref{sec:best} we characterize the class of games where the VE is efficient and show that the GNE can be arbitrarily
inefficient in this class. Section~\ref{sec:remedy} some ways to lower bound the efficiency loss and some other models are considered. 
We conclude with some final considerations in Section~\ref{sec:final}.

\section{Preliminaries} \label{sec:preliminaries}
In this section we describe the setting for our resource allocation game 
and our notions of efficiency. We mention some mathematical characterizations that will be used in later sections. 

Let $\Nscr = \{1,2,\hdots,N\}$ be a set of players. 
For each $i \in \Nscr$, let $x_i \in \Real$ be player $i$'s strategy and  $\varphi_i : \Real^N \rightarrow \Real$ be his 
utility function. By $x$ we denote the tuple $(x_1,x_2,\hdots,x_N)$,   
 $x^{-i}$ denotes the tuple $(x_1,\hdots,x_{i-1},x_{i+1}, \hdots, x_N)$ and $(y_i,x^{-i})$ 
the tuple $(x_1,\hdots,x_{i-1},y_i,x_{i+1},\hdots,x_N)$. 
Let $\Phi$ denote the tuple $(\varphi_1,\hdots, \varphi_N)$.
The {\em shared constraint} in our game 
is the  requirement that the tuple $x$ be constrained to lie in a set 
$\Cbb \subseteq \Real^N$,
$$ \Cbb = \left\{ x \geq 0\ \left\lvert \ \sum_{j \in \Nscr} x_j \leq C \right. \right\},$$  
where $C$ is a positive real number 

A tuple $x \in \Cbb$ is termed an {\em allocation} and the number $C$ represents the capacity of the resource. 
We consider an model with a single resource, whereby $x_i$ is the portion of the resource demanded by player $i$. 
In the generalized Nash game with shared constraint $\Cbb$, player $i$ is assumed
to solve the {parameterized} optimization problem, 
	$$
	\maxproblem{A$_i(x^{-i})$}
	{x_i}
	{\varphi_i(x_i;x^{-i})}
				 {\begin{array}{r@{\ }c@{\ }l}
		\sum_{j\in \Nscr} x_j &\leq& C, \\
		x_i & \geq& 0.
	\end{array}}
	$$
In other words, each player decides a portion of the resource so as to maximize his utility, but he does so while being 
cognizant of the fact that portion he has access to is that which when added to the portions of other players, does not 
exceed $C$. 
Throughout this paper we abbreviate the aggregate utility function by 
\begin{equation}
\Theta \triangleq \sum_{j \in \Nscr} \varphi_j, \label{eq:thetanote} 
\end{equation}
and make the following assumption.
\begin{assumption} 
We assume that for each $i \in \Nscr$, the utility function $\varphi_i(x)$ is a concave, continuously differentiable  
and strictly increasing function in $x_i$. Furthermore, the utility obtained from the allocation $x=0$, \ie $\varphi_i(0)$,  is nonnegative for all $i$. 
Finally, the aggregate utility $\Theta$ is a concave function such that 
for each allocation $x$, $\nabla \Theta(x)$ is component-wise nonnegative. \label{5-assump}
\end{assumption}

Following are the motivations behind these assumptions. It is common to assume that utility is an increasing 
function of the portion, and that {\em marginal utility} is nonincreasing. This leads one to an assumption that
of each $\varphi_i$ is increasing and strictly concave in $x_i$. The other assumptions, 
in particular the concavity of $\Theta$ and nonnegativity of $\nabla \Theta$  
may not hold in all cases. The nonnegativity of $\nabla \Theta$ and the concavity of $\Theta$ 
says that utility functions of
players are such that for the 
social planner, withholding portions of the resource so as to not exhaust the capacity $C$, is not optimal. 
% The concavity of 
% $\Theta$ is a technical assumption that helps simplify our analysis. 
We note that these assumptions are compatible wth the usual 
case where for all $i$, $\varphi_i(x) = U_i(x_i)$ for all $x$ and 
that they are met by quasi-linear utility functions: $\varphi_i(x) = U_i(x_i) - \sum_{j \neq i} d^i_j
x_j$ commonly used in economics~\cite{mas-colell95microeconomic}.

We recall the solution concepts of the generalized Nash equilibrium 
(GNE) and variational equilibrium (VE) from~\cite{kulkarni09refinement,pang09convex,facchinei07generalized}. 
For this game, these concepts reduce to the following.
\begin{definition}[Generalized Nash equilibrium (GNE)]
A tuple $x$ is a generalized Nash equilibrium if there exists a tuple $\Lambda = (\lambda_1,\hdots,\lambda_N) \in \Real^N$ such that 
$x$ and $\Lambda$ satisfy the KKT conditions of the optimization problems {\rm (A$_1),\hdots,($A$_N$)}:
\begin{align*}
0 \leq x_i &\perp -\nabla_i\varphi_i(x) + \lambda_i \geq 0 \\
0 \leq  \lambda_i & \perp C - \bfone^T x \geq 0, \qquad \forall \ i \in \Nscr.
\end{align*}
The set of all GNEs of this game with objective functions $\Phi$ is denoted by $\GNE(\Phi)$.
\end{definition}
Here $\bfone$ denotes a vector in $\Real^N$ with each coordinate unity. 
The GNE that corresponds to equal Lagrange
multipliers, \ie 
$\Lambda = \lambda \bfone$ for some $\lambda$ is the VE.
\begin{definition}[Variational equilibrium (VE)]
A tuple $x$ is a variational equilibrium if there exists $\lambda \in \Real$ such that $x$ satisfies the KKT 
conditions of all  the optimization problems {\rm (A$_1),\hdots,($A$_N$)} with 
$\lambda$ as the Lagrange multiplier. \ie
 \begin{align*}
0 \leq x_i &\perp -\nabla_i\varphi_i(x) + \lambda \geq 0 \\
0 \leq  \lambda & \perp C - \bfone^Tx \geq 0, \qquad \forall \ i \in \Nscr.  
 \end{align*}
 The set of all VEs of this game with objective functions $\Phi$ is denoted as $\VE(\Phi)$.
\end{definition}
As a consequence of the compactness of $\Cbb$, any game with utility
 functions satisfying Assumption~\ref{5-assump}, 
the set of VEs (and hence the set of GNEs) is nonempty. The proof of this 
fact can be found, \eg, in Rosen~\cite{rosen65existence}.

It was argued in~\cite{kulkarni09refinement} that the VE may be considered a refinement of the GNE. The main argument rests on the observation that the Lagrange multipliers $\lambda_1,\hdots,\lambda_N$ in the definition of the GNE 
can be interpreted as \textit{prices} charged by an administrator for consumption of the resource. As such any specific GNE can be realized by the choice of these multipliers. A GNE that is not a VE corresponds to the imposition of nonuniform prices across players by the administrator, whereas the VE corresponds to uniform prices. However, nonuniform prices necessarily require the administrator to be able to distinguish between the players. If players are anonymous, only uniform prices are feasible and as such the only GNE that can be meaningful is the VE.

The goal of this work is to study the efficiency of these solution concepts. As mentioned in the introduction, 
our notion of efficiency is the maximization of aggregate utility. 
\begin{definition}
A point $x$ is said to be efficient if it solves the following optimization problem
 	$$
	\maxproblem{{\rm SYS}}
	{x}
	{\Theta(x)}
				 {\begin{array}{r@{\ }c@{\ }l}
		\sum_{j\in \Nscr} x_j &\leq& C, \\
		x_j & \geq& 0, \qquad \forall \ j \in \Nscr. 
	\end{array}}
	$$
The efficiency of a point $x$ is defined as the ratio 
$\ds \frac{\Theta(x)}{\Theta(x^*)},$ where $x^*$ is a solution of \SYS.
\end{definition}
By Assumption~\ref{5-assump}, \SYS\  is a convex optimization problem. 
It follows that $x$ is efficient if there exists $\lambda \in \Real$ such that 
\begin{align*}
0 \leq x & \perp -\nabla \Theta(x) + \lambda {\bf 1} \geq 0 \\
0 \leq \lambda &\perp C - \sum_{j \in \Nscr} x_j \geq 0
\end{align*}

An alternative characterization of the GNEs and VEs is through the use of quasi-variational 
inequalities and variational inequalities respectively. 
To define  these objects, let 
$F: \Real^N \rightarrow \Real^N$ be the following function 
$$F(x) = -\pmat{\nabla_1 \varphi_1(x) \\ \vdots \\ \nabla_N \varphi_N(x)}, \quad \forall \ x \in \Real^N,$$
and $K$ be the set-valued map
 \begin{align}
 K(x) &:= \prod_{i \in \Nscr} K_i(x^{-i})  \quad \mbox{where} \quad K_i(x^{-i}):= \left\{ y_i \in \Real^{m_i} \ | \ (y_i, x^{-i}) \in  \Cbb\right\}, \quad \forall i \in \Nscr, \forall\ x \in \Real^N. \label{5-kxdef}
 \end{align}
A allocation $x$ is a GNE if and only if $x$ solves the quasi-variational inequality \QVI$(K,F)$ below.

\boxedeqnnew{Find $x \in K(x)$ \quad such that \quad  $F(x)^T (y - x) \geq 0$\quad $\forall \ y \in K(x)$.
          }{(\QVI$(K,F)$)}
Likewise, $x$ is a VE if and only if it solves the variational inequality \VI$(\Cbb,F)$.

\boxedeqnnew{Find $x \in \Cbb$ \quad such that  \quad  $F(x)^T (y - x) \geq 0$ \quad $\forall \ y \in \Cbb$.
          }{(\VI$(\Cbb,F)$)}
Since under Assumption~\ref{5-assump} \SYS\ is a convex optimization problem, efficient allocations 
are characterized by the solutions of the $\VI(\Cbb, -\nabla \Theta)$,

\boxedeqnsmall{Find $x \in \Cbb$ \quad such that  \quad  $-\nabla \Theta(x)^T (y - x) \geq 0$ \quad $\forall \ y \in \Cbb$.
          }{(\VI$(\Cbb,-\nabla\Theta)$)}{1}{.05}{.7}{.25}

The efficiency of an allocation depends on utility functions considered. 
In order to provide guarantees of efficiency of an allocation one needs to consider classes of utility functions 
and examine the worst (or best) case of efficiency over all of them. Therefore to provide efficiency guarantees for a {\em solution 
concept}, one needs to examine the worst (or best) case of efficiency over allocations generated by the said solution concept
for each utility function in the class. For this purpose denote by $\Fscr$ the class of utility functions that 
satisfy Assumption~\ref{5-assump}:
$$ \Fscr = \{ \Phi \ | \  \varphi_1,\hdots, \varphi_N \ \mbox{satisfy Assumption~\ref{5-assump}} \}. $$
Let $\Lscr$ be the subclass of $\Fscr$ comprising of {\em linear} objective functions:
$$ \Lscr = \{ \Phi  \ | \  \varphi_1,\hdots, \varphi_N \ \mbox{are linear and satisfy Assumption~\ref{5-assump}} \}. $$
Following are our notions of best case and worst case efficiency for the GNE and VE respectively. 
\begin{definition}
 The best case efficiency of the GNE and the VE are defined as 
\begin{align*} \bar{\rho} &= \sup_{x \in \GNE(\Phi),\: \Phi \in \Fscr} \frac{\Theta(x)}{ \max_{z \in \Cbb} \Theta(z)}, \\
\aur \quad \bar{\vartheta} &= \sup_{x \in \VE(\Phi),\: \Phi \in \Fscr} \frac{\Theta(x)}{\max_{z\in \Cbb}\Theta(z)},
\end{align*}
respectively.
\end{definition}
\begin{definition}
 The worst case efficiency of the GNE and the VE are defined as 
\begin{align*} \underline{\rho} &= \inf_{x \in \GNE(\Phi),\: \Phi \in \Fscr} 
\frac{\Theta(x)}{ \max_{z \in \Cbb} \Theta(z)}, \\
\aur \quad \underline{\vartheta} &= \inf_{x \in \VE(\Phi),\: \Phi \in \Fscr} \frac{\Theta(x)}{\max_{z\in \Cbb}\Theta(z)},
\end{align*}
respectively.
\end{definition}

By Assumption~\ref{5-assump}, $\Theta(0) \geq 0$ and $\nabla \Theta \geq0$. This ensures that $\Theta$ is nonnegative on
$\Cbb$ 
and thereby the solution of \SYS\ is positive and finite, whereby the efficiencies defined above are all finite and nonnegative. 
In other words, $\bar{\rho}, \underline{\rho}, \bar{\vartheta}, 
\underline{\vartheta} \in [0,1],$ and that $\bar{\rho} \geq \underline{\rho}$ and $\bar{\vartheta} \geq \underline{\vartheta}$.   
Furthermore, since every VE is a \GNE, we readily get 
$$\bar{\rho} \geq \bar{\vartheta}, \qquad \mbox{and} \qquad \underline{\rho} \leq \underline{\vartheta}.$$

\section{Relation to past work} \label{sec:relation}
The appropriate allocation of resources is perhaps the fundamental concern that led to the inquiries that 
we today recognize as being part of the field of economics. Adam Smith's 
classic, {\em The Wealth of Nations}~\cite{smith08inquiry}, first published back in 
1776, contains perhaps the first scholarly attempt at understanding and explaining how societies allocate resources. 
Smith's observations, which say that competition and efficiency go hand in hand, are the foundations of 
{\em general equilibrium theory} and the {\em welfare theorems}~\cite{mas-colell95microeconomic,walras52elements}: 
every competitive equilibrium maximizes social welfare and 
under certain conditions, every social welfare maximizing allocation is achievable as a competitive equilibrium. 

These theorems rely on the assumption of {\em perfect competition} which states that players 
do not strategically anticipate the impact of their actions on the actions of other players~\cite{mas-colell80noncooperative}. 
Things changed in the mid-1940's with the invention of game theory~\cite{neumann44games}. Game theory provided 
a formal paradigm for studying the more general and, in some cases, somewhat more realistic setting where 
players were allowed to be strategic and {\em price anticipating}. In game-theoretic parlance, social welfare maximization was 
seen as the outcome of a {\em cooperative} game. By considering the Nash equilibrium as the outcome of the noncooperative 
counterpart, it was easily seen that the welfare theorems need not hold, \ie cooperative and noncooperative games yield 
patently different equilibria. The classic game of {\em Prisoner's Dilemma}~\cite{mas-colell95microeconomic} 
is a telling example of this fact. The noncooperative interaction yields both prisoners significantly more years in prison than a 
decision they could have achieved had they cooperated. In fact a  theorem of Dubey~\cite{dubey86inefficiency} rigorously establishes that 
Nash equilibria are generically (for an open dense subset of the space of 
utility functions) inefficient in the Pareto sense\footnote{See also Maskin~\cite{maskin99nash} which gives sufficient conditions 
for a social choice rule to implementable by a game and the Pigou example~\cite{nisan07algorithmic}}.

Efficiency has also been a question in mechanism design. Though Nash equilibria can be inefficient, in the context of mechanism design, 
these Nash equilibria correspond to the game in the space of bids induced by the mechanism. In this setting, 
more encouraging results are seen. The Vickrey auction~\cite{vickrey61counterspeculation}, is an success story of this approach, for it 
is known to be efficient\footnote{The Vickrey auction pertains to the auctioning of a {\em single} indivisible good. Efficiency  
in this context means that the good is allocated to the bidder who values it most.} and has other attractive properties. The last two decades in particular has seen a dramatic surge in 
the interest in such questions, especially in the computer science and electrical engineering community. 
Most of this interest find its origins in the seminal work of Kelly \etal~\cite{kelly98rate} in which 
congestion or rate control in a communication network, such as the internet, 
was modelled as a resource allocation problem. Kelly \etal introduced a model in which each player 
submitted a ``bid'' and received an allocation proportional to its bid and showed that the competitive equilibrium 
of this mechanism was efficient. Since then Johari~\cite{johari04thesis}, Roughgarden 
\cite{roughgarden02selfish,roughgarden04bounding}, Papadimitriou~\cite{papadimitriou01algorithms} and many others 
have worked on  bounding the efficiency loss in resource allocation or related games. Indeed, Papadimitriou 
has called the worst case efficiency as the {\em price of anarchy}. The volume of work 
in this area is too large to be cited here to any degree of completeness. We instead refer the reader to the above references 
and to the book by Nisan~\etal~\cite{nisan07algorithmic} for more. 

Our work is distinct from all the other work on resource allocation. But thanks to this long and rich history of 
resource allocation, it does bear similarities to some of the past approaches. We highlight the distinctions here and 
the relation other approaches 
is explained in the following sections. The first distinction between our approach and any of the past approaches is the use 
of a shared-constraint game. 
Though shared-constraint game formulation appeals naturally to resource allocation, to the best of 
our knowledge there is no work directly shared-constraint games as a vehicle for the abstract question of 
resource allocation. 
There is however work in the specific context of congestion in networks~\cite{alpcan02game}. Secondly, our work does not 
attempt to provide an alternative to mechanisms, but instead considers the setting where the option of a mechanism is not 
available. In our setting, players wish to split a pie $N$-ways without communication with each other and without the intervention 
of an administrator while simultaneously moving to get pieces of it. This is perhaps an example of a somewhat  
extreme {\em anarchy} in resource allocation. Finally, the shared-constraint game admits to kinds of equilibria (GNE and VE) and 
with different consequences on efficiency. We want to know how they perform in the context of efficiency. This 
question has not been studied before.

The varied nature of the equilibria of the shared-constraint game has the potential to yield extremely high efficiency as well 
as extremely low efficiency and therein lies the possibility for similarities with other approaches. 
With utility functions of the ``perfectly competitive'' kind, \ie the utility obtained by a player is a function only 
of the allocation received by it, the VE is seen to be identical to the competitive equilibrium. We elaborate on this 
comparison in Section~\ref{sec:competitive}. In Section~\ref{sec:mechanism}, we compare our approach for this same 
model with a typical mechanism-based approach. 

% Having said this, as a corollary of our work, we obtain settings where 
% allocating resources through a shared-constraint game may be more efficient than allocating them using a mechanism. 

\subsection{Relation to the competitive equilibrium} \label{sec:competitive}
Consider a setting with one {\em seller} and $N$ buyers and a single resource with capacity $C$. 
Buyers are similar to what we have called {\em players} 
in this work; they are characterized by utility functions $\varphi_i(x) = U_i(x_i)$ for all $x$ and $i \in \Nscr$, 
which we assume satisfy Assumption~\ref{5-assump}. 
Imagine that each buyer sees a unit price $p$ for the resource. Taking this price for granted, the buyer decides
a quantity to buy so as to maximize his net payoff  which is his utility less the payment he makes. Thus buyer $i$ is faced with the following optimization problem.
$$
        \maxproblem{B$_i$}
        {x_i}
        {U_i(x_i) -px_i}
                                 {\begin{array}{r@{\ }c@{\ }l}
                x_i & \geq& 0.
        \end{array}}
        $$
The seller, on the other hand, seeks to maximize revenue. The sellers decision variables are the quantity sold 
to each buyer and he is constrained by the capacity of the resource. The seller also takes this price for granted 
in making his decision and is therefore faced with the following optimization problem.
$$
        \maxproblem{S}
        {x}
        {p(\bfone^Tx)}
                                 {\begin{array}{r@{\ }c@{\ }l}
                \bfone^Tx & \leq& C,\\
                  x &\geq & 0.
        \end{array}}
        $$
Now imagine a societal objective of maximizing the aggregate utility of all buyers\footnote{This is also the {\em aggregate surplus} 
of buyers and sellers. See~\cite{johari04thesis} for more on this.}, which corresponds to the following problem.
$$
        \maxproblem{SOC}
        {x}
        {\sum_{i \in \Nscr}  U_i(x_i)}
                                 {\begin{array}{r@{\ }c@{\ }l}
                \bfone^Tx & \leq& C,\\
                  x &\geq & 0.
        \end{array}}
        $$
The setting described above is called the perfectly competitive setting. Notice that buyers and sellers {\em simultaneously} 
the quantities $x_1,\hdots,x_N$. The ``price'' is assumed to be exogenously given\footnote{Akin to Adam Smith's {\em Invisible Hand}} and in addition there is a societal objective maximizing social welfare with the same decision variables. The first welfare theorem states that there exists a price $p$ such that all the above optimzation problems are consistently solved. \ie 
there exists a price $p$ and an allocation $x$ so that $x_i \in \SOL({\rm B}_i)$, for all $i \in \Nscr$,  
$x \in \SOL({\rm S})$ and $x \in \SOL({\rm SOC})$. 
The evidence for this lies in the fact that the KKT conditions of the problem (SOC) 
which say that $x$ solves (SOC) if there exists a $\lambda \in \Real$ such that 
\begin{align*}
 0 \leq x_i &\perp -\nabla_i U_i(x_i) + \lambda \geq 0,  \\
 0 \leq \lambda & \perp C - \bfone^T x \geq 0, \quad \forall  \ i \in \Nscr,
\end{align*}
 are identical to the KKT conditions of (B$_1),\hdots,$(B$_N$) and (S) taken together with $p = \lambda$. 
 This $x$ is called the {\em competitive equilibrium}\footnote{Oftentimes the tuple $x,p$ is called the competitive equilibrium}, 
 and the first welfare theorem states that the competitive equilibrium is efficient.

Now consider a shared constraint game with these buyers as players. Player $i$ solves the problem 
$$
        \maxproblem{A$_i(x\mi)$}
        {x_i}
        {U_i(x_i)}
                                 {\begin{array}{r@{\ }c@{\ }l}
                \bfone^Tx & \leq& C, \\
                x_i & \geq & 0.
        \end{array}}
        $$
Consider the solution concept of the VE. 
The allocation $x$ is a VE for this game if and only if there exists $\lambda$ 
such that 
\begin{align*}
 0 \leq x_i &\perp - \nabla U_i(x_i) + \lambda  \geq 0, \\
 0 \leq \lambda & \perp C - \bfone^T x \geq 0, \quad \forall  \ i \in \Nscr.
\end{align*}
Herein lies the connection between the competitive equilibrium and our model. For utility functions of the kind considered 
in perfectly competitive settings, the solution concept of the VE provides an allocation 
identical to that of a competitive equilibrium and is thereby efficient. Thus one may say 
that at a VE, players demand resources {\em as if} they were sold by a seller with a fixed price.

\subsection{Relation to mechanism design}\label{sec:mechanism}
Resource allocation through a shared-constraint game is different from that through 
the use of mechanisms. To clarify this difference we compare our approach with the approach of Johari~\cite{johari04thesis,johari04efficiency}, which 
may be taken to be a canonical mechanism-based approach to resource allocation.

\begin{examplec} \label{ex:johari}
Johari considers the utility functions from the perfectly competitive setting 
$\varphi_i(x) = U_i(x_i)$ for all $x$. Every player submits a {\em bid}, or willingness to pay to a system administrator. 
Let the bids of the players be $w_1,\hdots,w_N$ where each $w_i \in [0,\infty)$. The system administrator aggregates these bids and allocates a portion of the resource according to an allocation rule. For \eg, Johari uses the 
{\em proportional} allocation rule wherein player $i$ gets a portion $x_i$
$$ x_i = \frac{w_i}{\sum_{j \in \Nscr} w_j} C.$$
He is charged a payment $w_i$. The price for a unit of the resource is $\sum_{j \in \Nscr} w_j /C$. 
When players are price taking they choose their bids taking $p = \sum_{j\in \Nscr}w_j/C$ as constant 
and receive a quantity $x_i = w_i/p$. It is easy to show that this is allocation of the competitive equilibrium~\cite{kelly97charging}.

When players are price-anticipating, \ie they strategically anticipate the influence of their bids and 
their opponent's bids on the price, their interaction can be modelled as a game. 
	$$
	\maxproblem{J$_i(w\mi)$}
	{w_i}
	{U_i\left(\frac{w_i}{\sum_{j\in \Nscr}w_j}C\right) -w_i}
				 {\begin{array}{r@{\ }c@{\ }l}
		w_i & \geq& 0.
	\end{array}}
	$$
Let $w^* = (w_1^*,\hdots,w_N^*)$ be the (unique) Nash equilibrium of this game and let $x^* =w^*\frac{C}{\sum_{j\in\Nscr}w_j^*}$. 
The worst case efficiency of the proportional allocation mechanism is the ratio 
$$\rho = \inf_{\Phi \in \Fscr} \frac{\Theta(x^*)}{\max_{z\in \Cbb}\Theta(z)}.$$
It was shown in~\cite{johari04thesis} that this ratio is $3/4$.

Observe that the use of a mechanism effectively alters the game. The strategies of the players are now their bids, 
the feasibility of the allocation, which in a competitive equilibrium is determined by the seller is now in the hands 
of the administrator who aggregates these bids. Furthermore, the efficiency 
claimed holds only for the particular mechanism used. 

In comparison, we directly consider the generalized Nash game over the space of allocations. As shown 
in Section~\ref{sec:competitive}, the VE for a game with these objective functions has (worst case) efficiency one.  
\end{examplec}

\section{Worst Case Efficiency}\label{sec:general}
In this section we derive the worst case efficiencies of the solution concepts of the GNE and the VE over the class of utility functions $\Fscr$. 

In the case of the GNE as well as the VE, we show that the worst case efficiency is achieved for the class of linear objective functions $\Lscr$. The worst case efficiency is then calculated 
constructively. This part of the analysis that uses estimation with linear objective functions 
follows lines similar to those in~\cite{johari04efficiency}. In particular, we require the following lemma, 
which is similar to the one used in~\cite{johari04efficiency}.
\begin{lemma} \label{lem:linear}
Suppose Assumption~\ref{5-assump} holds. Let $\Phi \in \Fscr$ and let $x^*$ a solution of \SYS. For every $x\in \Cbb$, we have 
\begin{equation}
\frac{\Theta(x)}{\Theta(x^*)} \geq \frac{\nabla \Theta(x)^Tx}{\nabla \Theta(x)^Tx^\ell},   \label{eq:linear}
\end{equation}
where $x^\ell$ solves the problem \SYS\ linearized at $x$, {\rm (SYS$^\ell(x)$)}:
  	$$
	\maxproblem{{\rm SYS$^\ell(x)$}}
	{z}
	{\nabla \Theta(x)^Tz}
				 {\begin{array}{r@{\ }c@{\ }l}
	z &\in & \Cbb.
\end{array}}
	$$
\end{lemma}
The proof is in the Appendix.
\subsection{Worst case efficiency of the VE} \label{sec:wcve}
% The previous section showed that the GNE can be arbitrarily inefficient. A more interesting question is the worst case
% efficiency of the \VE. 
Suppose we allow $\Phi$ to be any function in $\Fscr$, \ie, $\varphi_1,\hdots,\varphi_N$ are  utility functions 
satisfying Assumption~\ref{5-assump}. In this section we ask the following question: 
what is the worst possible efficiency  that the GNE and the VE can yield when $\Phi$ varies over the class $\Fscr$? In other words, what are 
the ratios $\underline{\vartheta}$ and $\underline{\rho}$? We prove that the worst case efficiency of the 
VE over is in fact zero, \ie $\underline{\vartheta} =0$. It follows  that $\underline{\rho}=0$. 

Our argument proceeds by characterizing $\underline{\vartheta}$. Theorem~\ref{thm:lscrprime} shows that the worst case efficiency of 
the VE over $\Fscr$ in fact the same as the worst case efficiency over $\Lscr$.  To prove this, let us set 
up some notation. For arbitrary $\bar{x} \in \Cbb$ and $\Phi \in \Fscr$, define the linearized 
functions $\widetilde{\varphi}^{\bar{x}}_i, i \in \Nscr$ and $\widetilde{\Theta}^{\bar{x}}$ as 
$$ \widetilde{\varphi}_i^{\bar{x}}(x) = \nabla \varphi_i(\bar{x})^Tx, \ i \in \Nscr \quad \aur \quad
\widetilde{\Theta}^{\bar{x}}(x)= \nabla \Theta(\bar{x})^Tx = \sum_{j\in \Nscr} \widetilde{\varphi}_j^{\bar{x}}(x).$$ 
% The problem \SYS$^\ell(\bar{x})$ is the same as the problem \SYS\ with $\Theta$ replaced by $\widetilde{\Theta}$. 
% The above Lemma shows that for all functions $\Theta$ and for all $x\in \Cbb$ there exists a linear function 
% $\widetilde{\Theta}$ such that the value $ \frac{\Theta(x)}{\max_{z\in \Cbb} \Theta(z)} $ dominates the value 
% $\frac{ \widetilde{\Theta}(x)}{\max_{z\in \Cbb} \widetilde{\Theta}(z)}$. 
Clearly the utility functions $\widetilde{\Phi}^{\bar{x}} :=(\widetilde{\varphi}_1^{\bar{x}},\hdots,\widetilde{\varphi}_N^{\bar{x}})$ lie 
in the class $\Lscr$. Let $\widetilde{F}^{\bar{x}}$ denote the mapping 
$$\widetilde{F}^{\bar{x}}(x) = -\pmat{\nabla_1 \widetilde{\varphi}_1^{\bar{x}}(x) \\ \vdots 
\\ \nabla_N \widetilde{\varphi}_N^{\bar{x}}(x) }. $$
Observe that 
$$ \widetilde{F}^{\bar{x}}(x)= -\pmat{\nabla_1 \varphi_1(\bar{x}) \\ \vdots \\ \nabla_N \varphi_N(\bar{x})} = F(\bar{x}) \quad \forall x. $$ 
  This also shows that every VE of a game with $\Phi \in \Fscr$ 
is a VE of some game with linear objectives.
\begin{lemma} \label{lem:subset}
Let $\Phi \in \Fscr$. Then 
$$ \VE(\Phi) \subseteq \bigcup_{x \in \VE(\Phi)} \VE(\widetilde{\Phi}^{x}).$$
\end{lemma}
\begin{proof}
Since $\bar{x}$ is a VE for objective functions $\Phi$, 
$$ F(\bar{x})^T(y-\bar{x}) \geq 0 \quad \forall \ y \in \Cbb \implies 
\widetilde{F}^{\bar{x}}(\bar{x})^T(y-\bar{x}) \geq 0 \quad \forall \ y \in \Cbb.$$ 
This in turn implies $\bar{x} \in  \VE(\widetilde{\Phi}^{\bar{x}})$.
\end{proof}

% Therefore for all $\Phi \in \Fscr$ and all $\bar{x} \in \VE(\Phi)$ there exists $\widetilde{\Phi} := \widetilde{\Phi}^{\bar{x}} \in \Lscr$ 
% such that $\bar{x}\in \VE(\widetilde{\Phi})$. So, $\bigcup_{\Phi \in \Fscr}\VE(\Phi) \subseteq \bigcup_{\Phi \in \Lscr} \VE(\Phi)$. 
% But since $\Lscr \subseteq \Fscr$, 
% $$ \bigcup_{\Phi \in \Fscr}\VE(\Phi) = \bigcup_{\Phi \in \Lscr} \VE(\Phi).$$ 
Following the central result of this section. We show that for the worst case efficiency of the VE over the class $\Fscr$, $\underline{\vartheta}$, it suffices to look at the class of linear objective functions, $\Lscr$. The proof is in the appendix.
\begin{theorem} \label{thm:lscrprime}
Suppose Assumption~\ref{5-assump} holds. Then 
$$\underline{\vartheta} = \inf_{\bar{x} \in \VE(\Phi),\: \Phi \in \Fscr}  \frac{\Theta(\bar{x})}{\max_{z \in \Cbb}\Theta(z)}
= \inf_{\bar{x} \in \VE(\Phi),\: \Phi \in \Lscr}  \frac{\Theta(\bar{x})}{\max_{z \in \Cbb}\Theta(z)}.$$
\end{theorem}

The above result does not provide us the value of $\underline{\vartheta}$ but only a characterization of it. 
We now show that $\underline{\vartheta}=0$ through an example.

\begin{examplec} \label{ex:wcve}
Consider a game with objective functions $\varphi_i(x) = d_i^Tx$, where $d_i = (d^1_i,\hdots,d^N_i)$ $\in \Real^N$ such that
$\nabla \Theta = \sum d_i \geq 0$. Furthermore, assume that $d_i^i > 0$ for all $i \in \Nscr$. It is easy to see this
collection of utility functions  
$\Phi$ satisfies Assumption~\ref{5-assump}. 

Let $c_i = d^i_i$ and $c := (c_1 , \hdots ,c_N)^T= F(x)$, for each $x$. 
The set of VEs of this game, $\VE(c)$ is the set of $x$ for which there exists $\lambda$ such that 
\begin{align*}
 0  \leq x & \perp -c + \lambda \bfone\geq 0 \\
 0 \leq \lambda & \perp C - \bfone^T x\geq 0
\end{align*}
Since $c>0$, observe that any $\lambda$ satisfying these equilibrium conditions must be strictly positive 
and that $\bfone^T x =C$ must hold for any VE $x$. 

We construct values for $d_1,\hdots,d_N$ so that there exists a VE whose worst case efficiency is arbitrarily close to zero. 
Let $\varepsilon \in (0,1)$ and let 
$$d_i^j := \begin{cases}%\frac{2}{N-1} & j =1\\
0, & \forall \ j \neq i, \\
       \varepsilon,    & j =i.
        \end{cases} \ \forall i \in \Nscr \backslash \{N\}, 
%        \quad d_1^j :=\begin{cases}
%                                                                               1 & \forall \ j \neq 1, \\
%                                                                               \varepsilon & j=1, 
%                                                                              \end{cases} 
                                                                              \quad d_N^j := \begin{cases}1 & j=1,\\
                                                                                                          0 & \forall \ j \notin \{N,1\},\\
                                                                                                          2 \varepsilon & j = N.
                                                                                                         \end{cases}
 $$
It follows that $c = (\varepsilon,\hdots,\varepsilon,2\varepsilon)$ and $d = \sum_{i\in \Nscr} d_i = (\varepsilon+1,\varepsilon,\hdots,\varepsilon,2\varepsilon)$. 
% Consequently, if $x$ is a VE of this game, $c^Tx = \varepsilon \bfone^Tx = C\varepsilon$ 
% and for any $z$, $d^Tz = (N-1 + \varepsilon) \bfone^Tz$. 
It is easy to check that $x^* := \left(0, \hdots,0, C\right)$ is a VE (the equilibrium conditions
are satisfied for $\lambda = 2\varepsilon$). Since $\varepsilon +1 > 2\epsilon$, the optimal value of \SYS\  is 
$C(\varepsilon+1)$. 
The worst case efficiency of the VE is bounded above by the efficiency of $x^*$. Therefore, 
$$ \underline{\vartheta} \leq \frac{c^Tx^*}{\max_{z \in \Cbb} d^Tz} = \frac{2C\varepsilon}{C(\varepsilon+1)} =\frac{2\varepsilon}{\varepsilon + 1}.$$
Letting $\varepsilon$ decrease to zero reveals that  the worst case efficiency  $\underline{\vartheta} =0$,  
implying that the VE can be arbitrarily inefficient. Letting $\varepsilon$ approach one shows that 
efficiencies arbitrarily close to unity are also achievable in the class $\Fscr$. 
\end{examplec}\\
Summarizing the above example we have the following result for the worst case efficiency of the VE.
\begin{theorem}\label{thm:wcve}
 The worst case efficiency of the VE (and the GNE) 
 over the class $\Fscr$ is zero. This bound is tight in the sense that for any $\varepsilon \in (0,1]$, 
 there exists a game with objective functions $\Phi$ in $\Fscr$ such that a VE (and hence a GNE) of this game has efficiency $\varepsilon$.
\end{theorem}
\begin{proof}
We only need to prove the tightness. That efficiency of $\varepsilon$ for any $\varepsilon \in (0,1)$ 
are achieved is demonstrated in the example above. Section~\ref{sec:competitive} showed that in the setting of 
perfect competition, the VE coincides with the competitive equilibrium and is therefore efficient.  Thus efficiency $\varepsilon =1$
is also achievable. 
\end{proof}

\section{Best Case Efficiency} \label{sec:best}
In this section we present the {\em best} case analysis of the efficiency of the GNE and the VE. Although  
Section~\ref{sec:competitive} demonstrates that the best case efficiency of the VE and the GNE is 
unity, we go a step further in this section and characterize the class of games for which every VE is efficient. 
\ie, in this case, the best {\em and} the worst case efficiency of the VE is unity. 
This class is large and it includes in it the perfectly competitive setting. 
Then, restricting ourselves to this class we determine  the worst case efficiency of the GNE over this class. 
We find that while the VE is efficient for every game in this class, the GNE can be arbitrarily inefficient.

For the VE to be efficient, it has to be optimal for the problem \SYS. To find a sufficient condition for 
the VE to solve \SYS, let us first answer a related question: is there an optimization problem that the VE solves?
By this we seek another optimization problem whose optimality conditions are the same as the equilibrium 
conditions of the \VE, whereby solving the two is equivalent. 
Fortunately the answer to this is rather simple. The VE is the solution of \VI($\Cbb,F$). If there exists a
 concave differentiable function $f : \Real^N \rightarrow \Real $ such that 
\begin{equation}
-\nabla f(x) =  F(x) \quad \forall \ x \in \Real^N, \label{eq:fdef}
\end{equation}
then \VI($\Cbb,F$) is equivalent to the \VI($\Cbb,-\nabla f$) which, by the convexity of $\Cbb$ is equivalent to the optimization problem
 	$$
	\maxproblem{\VE-OPT}
	{x}
	{f(x)}
				 {\begin{array}{r@{\ }c@{\ }l}
		\sum_{j\in \Nscr} x_j &\leq& C, \\
		x& \geq& 0.
	\end{array}}
	$$
One way of interpreting \eqref{eq:fdef} is that it asks for $F$ to be integrable and for $-f$ to be its integrand. 
It is well known that such a function $f$ exists if and only if the Jacobian $\nabla F(x)$ is symmetric~\cite{ortega87iterative}. 

Now, suppose such a function does exist. Then since $-\nabla f(x) = F(x)$ holds for all $x$, we must have for each $i$ \begin{align}
\nabla_i \left[ f(x) - \varphi_i(x)\right] &=0 \non \\
\mbox{\ie,} \quad f(x) - \varphi_i(x) &= \eta_i(x^{-i}), \label{eq:1}
\end{align}
for some $\eta_i$ which is a function of $x\mi$ and for all $i\in \Nscr$. Using this, one may now ask the original question. For what objective functions 
$\Phi$ is $f(x) = \Theta(x) = \sum_{j \in \Nscr} \varphi_j(x)$? The following theorem 
provides us such a characterization. 

\begin{theorem}\label{thm:veeff}
Suppose Assumption~\ref{5-assump} holds. The identity $-\nabla \Theta =F$ holds if and only if there exist continuously
differentiable functions $\eta_i$ of $x\mi$ for each $i \in \Nscr,$ such that the utility functions are given by 
\begin{equation}
\varphi_i(x)= \frac{\sum_{j=1}^N \eta_j(x^{-j})}{N-1} - \eta_i(x\mi), \label{eq:slade}
\end{equation}
for every $i \in \Nscr$.
\end{theorem}
\begin{proof}
``$\implies$'' Let  $-\nabla \Theta = F$. In particular, we have, for each $i\in \Nscr$, and for all $x$, 
$ \nabla_i \Theta(x) = \nabla_i\varphi_i(x).$ Recall from  Assumption~\ref{5-assump} that $\Theta$ is concave. Therefore
\eqref{eq:1} gives that for each $i \in \Nscr$,
\begin{align}
\Theta(x) - \varphi_i(x) &= \eta_i(x^{-i}), \label{5-eq:2}
\end{align}
for some $\eta_i$ which is a function of $x^{-i}$. Since $\varphi_i$ and $\Theta$ are continuously differentiable, it follows
that 
$\eta_i$ is continuously differentiable. Summing \eqref{5-eq:2} over $i$ and using that $\Theta = \sum \varphi_i$, 
gives 
\begin{equation}
  \Theta(x) = \frac{\sum_{i=1}^N \eta_i(x^{-i})}{N-1}. \label{eq:thetadef}
\end{equation}
Now substituting for $\Theta$ from \eqref{eq:thetadef} in \eqref{5-eq:2} gives the result desired in \eqref{eq:slade}. 

``$\impliedby$'' Suppose \eqref{eq:slade} holds for some continuously differentiable functions $\eta_1,\hdots,\eta_N$. Summing \eqref{eq:slade} over $i$ gives \eqref{eq:thetadef}. 
Then for each $i \in \Nscr,$
$$\nabla_i \varphi_i(x) = \nabla_i \left( \frac{\sum_{j=1}^N \eta_j(x^{-j})}{N-1}\right) = \nabla_i\Theta(x). $$
In other words, $-\nabla \Theta = F$.
\end{proof}

Theorem~\ref{thm:veeff} shows that for a class of objective functions, the worst and best 
case efficiency of the VE is unity.
\begin{theorem}
Suppose Assumption~\ref{5-assump} holds. Then the best case efficiencies of the VE and GNE are unity. \ie $\bar{\rho} = \bar{\vartheta} =1$.
\end{theorem}

Remarkably, the class given by \eqref{eq:slade} does not depend on $\Cbb$. \ie if $\varphi_1,\hdots,\varphi_N$ 
are of the form given by \eqref{eq:slade}, then the \VI($V,F$) is equivalent to \VI$(V,-\nabla \Theta)$ for any closed convex
set $V$. 
Note also that from \eqref{eq:1},
$$\eta_i(x\mi) = \sum_{j\in \Nscr} \varphi_j(x) - \varphi_i(x) = \sum_{j \neq i} \varphi_j(x). $$ 
In general the right hand side may not independent of $x_i$, and the above equation may not hold. Thus, 
not every game has VEs equivalent to optimization problems. 

We now revisit the relationship of the VE with the competitive equilibrium in the light of this result. 

\begin{examplec}
In Section~\ref{sec:competitive}, we considered a perfectly competitive setting and then we considered a \scg analogous 
to it. The game considered therein was with $\varphi_i(x) = U_i(x_i)$ for all $i$. For this game, we indeed have 
the functions $\eta_i$ given by 
$$\eta_i(x\mi) = \sum_{j\neq i}U_j(x_j), $$
which is independent of $x_i$.
\end{examplec}

Notice that the statement `$-\nabla \Theta = F$' characterized by Theorem~\ref{thm:veeff} is  somewhat stronger than 
the statement `every VE is efficient'. `$-\nabla \Theta =F$' provides the equivalence  between the VIs characterizing the VE and solutions of \SYS. This equivalence implies the equality of the solution 
sets of the said \VIs, though it is not necessary to conclude the equality of these sets. Following is 
an example of a game where objectives do not satisfy 
\eqref{eq:slade}, every VE is efficient.
\begin{examplec}
Consider $\varphi_i(x) = x_i g(\bfone^Tx)$ whereby $\Theta(x) = \bfone^T x g(\bfone^T x) $ and assume that $g(\bfone^Tx) + x_ig'(\bfone^Tx) >0$ and 
$\nabla \Theta (x) = \bfone (g(\bfone^Tx) + (\bfone^Tx)g'(\bfone^Tx)) \geq 0$ for all $x$. It is easy to check that this set of
functions 
does not satisfy \eqref{eq:slade}. But, as we show below, every VE of this game is efficient. 

Let us first consider the system problem \SYS. Since $\nabla \Theta \geq 0$, the maximum in \SYS\ is attained 
at $\bfone^Tx =C$, whereby the optimal value of \SYS\ is $Cg(C)$. 
Now consider the \scg formed from these utilities. Thus $x$ is a VE of this game if and only if
\begin{align*}
0 \leq x &\perp -g(\bfone^Tx)\bfone - xg'(\bfone^Tx) + \lambda\bfone \geq 0 \\
0 \leq \lambda & \perp C - \bfone^Tx \geq 0,  
\end{align*}
for some $\lambda \in \Real$. Since $g(\bfone^Tx)\bfone + x g'(\bfone^Tx) >0$, $\lambda = 0$ does not satisfy 
these equations. Consequently, for each VE $x$, the equality $\bfone^Tx = C$ must hold. 
% Thus the above equilibrium conditions  reduce to the condition, 
% \begin{align*}
% 0 \leq x &\perp -g(C)\bfone - xg'(C) + \lambda\bfone \geq 0.
% \end{align*}
% Now suppose for $i\neq j$, $x_i =0$ but $x_j >0$. Then $\lambda = g(C) + x_jg'(C)$. Furthermore from the equilibrium conditions for 
% $x_i$, $\lambda \geq g(C)$. This implies that $x_jg'(C) \geq 0$. But $g$ is strictly decreasing, implying that $x_j \leq 0$. 
% That contradicts $x_j >0$. So we cannot have $x_i=0, x_j >0$ for $i\neq j$. Since $\bfone^Tx =C$, this implies that the 
% only alternative is $x_j >0$ for all $j$. In that case, $\lambda = \frac{1}{N}(g(C) + Cg'(C))$ and $x_1=\hdots=x_N = \frac{C}{N}$. 
% 
So for any VE $x$, $\Theta(x) = \bfone^Tx g(\bfone^Tx) = Cg(C)$, which is the optimal value of \SYS.
% 
% 
% Let $x^*$ be the solution of \SYS. $ \Theta(x) = \bfone^Tx g(\bfone^Tx)$, and so
% $$\nabla \Theta(x) = \bfone [g(\bfone^Tx) + (\bfone^Tx) g'(\bfone^Tx)]$$
% \begin{align*}
% 0 \leq x^* & \perp  - \nabla \Theta(x^*) + \lambda \bfone \geq 0 \\
% 0 \leq \lambda & \perp C - \bfone^Tx \geq 0
% \end{align*}
% 
% 
% $$\frac{\Theta(x)}{\Theta(x^*)} = C^Tg(C)$$
\end{examplec}

Also note that the requirement `$-\nabla \Theta = F$' is, strictly speaking, not the same as `$\VI(\Cbb,F) \equiv
\VI(\Cbb,-\nabla\Theta)$', 
because the solution set of a VI is invariant under multiplication of the function by a positive constant. 
Specifically, if $-c\nabla \Theta = F$, where $c>0$ is a real number, then VI$(\Cbb,F)$ is equivalent to VI$(\Cbb,F)$. 
Therefore if $F = -\nabla f$ and $f = c\Theta$, we get using \eqref{eq:1},
$$ c \Theta(x)  = \varphi_i(x) + \eta_i(x\mi).$$
Arguing as in Theorem~\ref{thm:veeff}, we get 
$$ \Theta(x) = \frac{\sum_{j \in \Nscr} \eta_j(x^{-j})}{cN-1} \quad \mbox{and} \quad \varphi_i(x) = \frac{\sum_{j\in \Nscr}
\eta_j(x^{-j})}{cN-1} - \eta_i(x\mi) \quad \forall \ x.$$

Finally, we note that \eqref{eq:1} has also appeared previously in literature in a different context. Slade in~\cite{slade94what} has derived \eqref{eq:1} as a means of giving sufficient conditions for the stationarity conditions of an unconstrained Nash equilibrium to be equivalent to an optimization problem. 

We will invoke this class of games in the following section. For reference, denote 
\begin{align} 
\Fscr' &= \{ \Phi \ | \ \Phi \in \Fscr \ \mbox{and} \ \exists \ \mbox{functions } \eta_i \ \mbox{of } x\mi, i\in \Nscr, \ \mbox{so that } \varphi_1,\hdots,\varphi_N \ \mbox{are given by \eqref{eq:slade}}  \}\non \\
\Lscr' &= \{ \Phi \ | \  \Phi \in \Lscr \ \mbox{and} \ \exists \ \mbox{functions } \eta_i \ \mbox{of } x\mi, i\in \Nscr, \ \mbox{so that } \varphi_1,\hdots,\varphi_N \ \mbox{are given by \eqref{eq:slade}}\}. \label{eq:fscrprime}
\end{align}
Denote the efficiencies of the GNE over the class $\Fscr'$ by 
\begin{align}
 \bar{\rho'} := \sup_{x \in \GNE(\Phi), \: \Phi \in \Fscr'} \frac{\Theta(\bar{x})}{\max_{z \in \Cbb}\Theta(z)} \quad
\mbox{and} \quad 
 \underline{\rho'} := \inf_{x \in \GNE(\Phi), \: \Phi \in \Fscr'} \frac{\Theta(\bar{x})}{\max_{z \in \Cbb}\Theta(z)}
\label{eq:rhoprime}
\end{align}
Clearly, $\bar{\rho'} =1$ and $\underline{\rho'} \geq \underline{\rho}$. The best and worst case efficiencies of the \VE, 
\begin{align*}
 \bar{\vartheta'} := \sup_{x \in \VE(\Phi), \: \Phi \in \Fscr'} \frac{\Theta(\bar{x})}{\max_{z \in \Cbb}\Theta(z)} \quad
\mbox{and} \quad 
 \underline{\vartheta'} := \inf_{x \in \VE(\Phi), \: \Phi \in \Fscr'} \frac{\Theta(\bar{x})}{\max_{z \in \Cbb}\Theta(z)}
\end{align*}
are both unity. 

% \section{Does there exist a GNE that is efficient when a VE is not?}
% In general, a GNE may be more or less efficient than a \VE. The reason is explained below.
% 
% A GNE solves $x^g \in \Cbb$
% $$ F(x^g)^T(y -x^g) \geq 0 \qquad \forall \ y \in K(x^g). $$
% A VE solves $$F(x^v)^T(y-x^v) \geq 0 \qquad \forall \ y \in \Cbb. $$
% An efficient point $x$ solves $$ \nabla \Theta(x)^T(y-x) \geq 0 \qquad \forall y \in \Cbb, $$
% where $\nabla \Theta = \nabla ( \sum \varphi_i ) = F + G$. 
% 
% Suppose $x^g$ is a GNE which is efficient. Then 
% $$ F(x^g)^T(y-x^g) + G(x^g)^T(y-x^g) \geq 0 \qquad \forall \ y \in \Cbb$$
% Or $F(x^g) + G(x^g) \in -\Nscr(x^g;\Cbb)$. But $F(x^g) \in -\Nscr(x^g;K(x^g))$, by virtue of being a \GNE.
% Let us take the specific case of $\Cbb = \{x \ | \ x\geq 0, {\bf 1}^Tx \leq C \}$. Suppose 
% $x^g \in \partial \Cbb$. Then $F(x^g) = - \Lambda$. We want 
% $$ F(x^g) + G(x^g) = -\lambda {\bf 1} \implies G(x^g) = \Lambda - \lambda {\bf 1}.$$

%\section{Worst case efficiency} \label{sec:worst}

In the following section the worst case efficiency of the GNE to the class of functions
$\Fscr'$ is addressed. \ie we calculate ratio $\underline{\rho}'$. We find that this ratio is in fact zero, 
indicating that the GNE can be arbitrarily inefficient even while the VE is efficient.

\subsection{Worst case efficiency of the GNE when the VE is efficient} \label{sec:wcgne}
% We show that the worst case efficiency of the GNE for functions in the class
%  $\Fscr'$ (\ie the class for which VE is efficient) is zero. \ie, $\underline{\rho'} =0$.
Recall the linearized objective functions from Section~\ref{sec:wcve} $\widetilde{\Phi}^{\bar{x}}$ where $\bar{x} \in \Cbb$ is the 
point of linearization. 
We first show that a result similar to Theorem~\ref{thm:lscrprime} holds, thanks to which $\underline{\rho'}$ 
is the worst case efficiency of GNE over the class $\Lscr'$. 
\begin{theorem} \label{thm:lscrprimeprime} 
Suppose Assumption~\ref{5-assump} holds. Then the ratio $\underline{\rho'}$ defined in \eqref{eq:rhoprime} is given by 
$$\underline{\rho'} =  \inf_{x \in \GNE(\Phi),\: \Phi \in \Lscr'} \frac{\Theta(x)}{\max_{z\in \Cbb} \Theta(z) }. $$
\end{theorem}

Once again we obtain that the linear case is where the worst case of efficiency is achieved. We now show through an example 
that this efficiency $\underline{\rho'}$ is indeed zero.

\begin{examplec} \label{ex:wcgne}
Consider a game with objective functions $\varphi_i(x) = d_i^Tx$ for each $i \in \Nscr$, where $d_i = (d^1_i,\hdots,d^N_i) \in \Real^N$ 
such that $d :=\nabla \Theta = \sum d_i \geq 0$ and $d_i^i > 0$ for all $i \in \Nscr$. 
It is easy to see this collection of objective function 
$\Phi$ satisfies Assumption~\ref{5-assump}. 
Let $c_i = d^i_i$ and $$c := \pmat{c_1 \\ \vdots \\ c_N}= \pmat{\nabla_1 \varphi_1 (x) \\ \vdots \\ \nabla_N \varphi_N(x)} = F(x), \qquad \forall\ x.$$ Furthermore we assume 
$$ \nabla \Theta = \sum d_i = c,$$
that is $\sum_{j \neq i} d_i^j = 0$ for all $i.$
Let $\GNE(c)$ denote the set of GNEs of this game. Recall that $x$ is a GNE if and only if there exists $\Lambda = (\lambda_1, \hdots,\lambda_N)^T$ such that 
\begin{align*}
 0  \leq x & \perp -c + \Lambda \geq 0, \\
 0 \leq \Lambda & \perp \bfone (C - \bfone^T x)\geq 0.
\end{align*}
Observe that any $\Lambda$ satisfying these equilibrium conditions must be strictly positive in all components 
because $c$ is assumed to be so. This implies that $\bfone^T x =C$ must hold for any GNE $x$. In fact, 
$\bfone^Tx = C$ along with $x\geq 0$ is also sufficient for $x$ to be a \GNE; the equilibrium conditions are then satisfied by $\Lambda =c$. 
So we get 
$$\GNE(c) = \{x \ | \ \bfone^Tx = C,\ x\geq 0 \}.$$
%$$ \rho = \min_{x \in \GNE(c),\: c \geq 0} \frac{c^Tx}{\max_{z \in \Cbb} c^Tz}.$$
% Suppose $x_i \in \Real$ and $\Cbb = \{ x \ | \ x\geq 0, \sum x_i \leq C \}$. Then for any $x \in \partial \Cbb$, 
% $$\mathscr{N}(x;K(x)) \subseteq  \{(\lambda_1, \hdots, \lambda_N)  \ | \ \lambda_1,\hdots , \lambda_N\geq 0\}. $$ 
% $x \in \GNE(c)$ iff $c \in \Nscr(x,K(x))$. It follows that for any $c \geq 0$, every point $x \in \partial \Cbb$ is a \GNE. 
Now let $1>\varepsilon >0$ and take $c =c_*:= (1,\varepsilon,\varepsilon,\hdots,\varepsilon)$. 
The point $x^* := \left(0, \hdots, 0, C\right)$ is a GNE for this $c$. 
The worst case efficiency of the GNE no greater than the efficiency of the GNE $x^*$ for the game with $c = c_*$. Therefore, 
$$ \underline{\rho'} \leq \frac{c^T_*x^*}{\max_{z \in \Cbb}c^T_*z} = \frac{C\varepsilon}{C} =\varepsilon.$$
Evidently, since $\varepsilon$ may take any value in $(0,1)$, the worst case efficiency $\underline{\rho'} =0$. 
In other words, the GNE can be arbitrarily inefficient. 
\end{examplec}

Combining the above example with the best case efficiency of the GNE (\ie $\varepsilon =1$), we get our final result. 
\begin{theorem}
Suppose Assumption~\ref{5-assump} holds. The worst case efficiencies of the GNE over the class $\Fscr$ and over class $\Fscr'$ are both zero. 
This bound is tight, in the sense that for every $\varepsilon \in (0,1]$ there is a game satisfying Assumption~\ref{5-assump} 
for which every VE is efficient but has a GNE that has efficiency $\varepsilon$.  
\end{theorem}

This is a particularly surprising result for it clearly shows the disparity between the VE and GNE. Example~\ref{ex:wcgne} also 
indicates the cause for loss of efficiency in the GNE. Games like these often admit a manifold of GNEs. In Example~\ref{ex:wcgne}, 
the set $\{ x\ | \ x\geq 0, \bfone^Tx = C\}$ is this manifold. When the utility functions are linear, it is always possible to find 
a GNE that allocates zero quantity the player with the highest contribution to $\Theta$ (\ie 
the player $i$ with the largest $\nabla_i \varphi_i$). One may also conclude that this, somewhat pathological property 
of the GNE, makes it unattractive as a solution concept for shared-constraint resource allocation games.

\section{Remedying zero worst case efficiency} \label{sec:remedy}
In this section we discuss some ways by which the possibility of arbitrarily low efficiency in the case of the GNE and the VE 
can be remedied.
\def\Fscrab{{\Fscr^{[\alpha,\beta]}}}
\subsection{Utility functions with bounded gradients}
Looking back at Example~\ref{ex:wcgne} we see that the limiting efficiency of zero is achieved 
as the gradient of $\Theta$ approaches the zero-vector. In the case of linear utilities, this corresponds to the possibility 
of constant zero utility. This is perhaps a pathological situation that does not occur in realistic settings. 
This observation also suggests a way for salvaging nonzero efficiency. 
Consider the following class of functions for $0<\alpha <\beta <\infty$,
$$\Fscr^{[\alpha,\beta]}  = \{ \Phi \ | \ \Phi \in \Fscr \ \mbox{and} \ \beta\bfone \geq \nabla \Theta(x) \geq \alpha\bfone \
\quad \forall\ 
x \in \Cbb \}. $$
Utility functions $\Phi$ belonging to this class have the gradients of their sum bounded above and below. 
%While for a given 
%continuously differentiable function $\Theta$, a finite bound on the gradient over the compact domain $\Cbb$ always exists. 
%$\Fscrab$ is a collection of functions, gradients of all of which are bounded by the same bound. 
Let $\Phi \in \Fscr$. Recall Lemma~\ref{lem:linear}, which showed that for any $x \in \Cbb$, 
the ratio $\frac{\Theta(x)}{\Theta(x^*)}$, where $x^* \in \SOL\SYS$ is was bounded below by the ratio $\frac{\nabla
\Theta(x)^Tx}{\nabla \Theta(x)^Tx^\ell}$ 
where $x^\ell$ solves (SYS$^\ell(x)$). 
Now if $\Phi \in \Fscrab$, then for any $x \in \Cbb$,
$$\frac{\Theta(x)}{\Theta(x^*)} \geq \frac{\alpha \bfone^Tx}{\beta \bfone^Tx^{\ell}}.$$
Furthermore since Assumption~\ref{5-assump} ensures that $F(x) >0$ whereby, for any GNE $x$, $\bfone^Tx =C$. 
Likewise, for the problem (SYS$^\ell(x)$) optimality of $x^\ell$ implies $\bfone^Tx^\ell =C$
This immediately gives us that 
$$\inf_{x \in \GNE(\Phi)} \frac{\Theta(x)}{\Theta(x^*)} \geq \frac{\alpha}{\beta}.$$
Since this holds for all $\Phi \in \Fscrab$, the worst case efficiency of the GNE over this class is $\frac{\alpha}{\beta}$. 
It follows that the worst case efficiency of the VE over the class $\Fscrab$ is greater than or equal to $\frac{\alpha}{\beta}$. 

\begin{examplec}
Consider the following class of {\em exponential} utility functions.
$$ \varphi_i(x) = 1-\exp(-d_i^Tx), \quad \forall  \ i \in \Nscr, $$
where for each $i$, the vector of coefficients $d_i=(d_i^1,\hdots,d_i^N)$ is chosen 
from a compact set $D \subseteq \Real^N_+$ and with $d_i^i >0$. It follows that these utility functions 
satisfy Assumption~\ref{5-assump}. 
It is easy to see that the following bounds hold.
$$ \alpha\sum d_i \leq \nabla \Theta(x) = \sum_{i \in \Nscr}  d_i\exp(-d_i^Tx) \leq \sum d_i, \qquad \forall \ x \in \Cbb,$$
where  $\alpha$ is given by 
$$\alpha = \min_{(d,z) \in D \times \Cbb} \exp(-d^Tz). $$ 
$\alpha$ is positive because $D$ is compact. Thus this set of utility functions lie in $\Fscr^{[\alpha,1]}.$
It follows that $\ds \frac{\Theta(x)}{\Theta(x^*)} \geq \alpha,$ 
and that the worst case efficiency of the GNE for this class of games is at least $\alpha$. 
\end{examplec}

\subsection{Other notions of efficiency}
The notion of what constitutes efficient allocations varies with the setting in question
and need not always correspond to the allocation that maximizes  the sum of the objective functions of players. 
This is particularly true when what players are maximizing are not merely their utilities but some 
modifications thereof. The system-level goal though, from the point of view of social welfare, remains the maximization of aggregate utility. 
As mentioned in Section~\ref{sec:mechanism}, mechanism design employs such a paradigm. 
The optimization problems that players solve are not of `utility maximization' but instead of 
`{\em payoff maximization}', though the efficiency is 
benchmarked on the aggregate utility. And under this interpretation, greater efficiency may be obtained.
\def\bfC{{\bf C}}

In this section, we study a particular kind of model that uses the above philosophy and provide a 
general principle for the VE to be efficient in this new sense. 
This model is adapted from and is a  somewhat simplified version of the model in 
\cite{alpcan02game}. We first explain the model from~\cite{alpcan02game}. Similar models are 
used in~\cite{kunniyur03endtoend,la00charge,mo00fair}.
\begin{examplec} \label{ex:alpcan}
The game consists of 
$N$ players that attempt to access bandwidth over a general network with a set of links $L$. 
The subset of links used by player $i$ are $L_i$ and the utility received by player 
$i$ is $U_i(x_i)$ for flow $x_i$ that it obtains. 
If $x$ is the vector of typical flows of all players, the 
flows permitted by the network are those in the set $\{ x \ | \ x \geq 0, \bfA x \leq \bfC\}$ where $\bfA$ is a $|L| \times N$ matrix whose
element $\bfA[\ell,i]$ is $1$ if player $i$ uses link $\ell$ and $\bfC$ is the $|L|$-vector of capacities the links.
The network is however subject to delay, the cost of which is seen by a player for every link it uses. 
Specifically, each player faces a cost 
which is the sum of the costs on the links that it uses,
$$c^i(x) = \sum_{\ell \in L_i} c_\ell( \sum_{j : \ell \in L_j }x_j),$$ 
where $c_\ell(t)$ is the cost of using link $\ell$ when the total flow on it is $t.$
Each player is thus faced with the following optimization problem.
$$
        \maxproblem{{\rm AB}$_i(x\mi)$}
        {x_i}
        {\varphi_i(x) := U_i(x_i) - c^i(x)}
                                 {\begin{array}{r@{\ }c@{\ }l}
                \bfA x &\leq& \bfC, \\
                x_i & \geq& 0. 
        \end{array}}
        $$
The game formed from these problems is clearly a \scg, more general than the game we have considered. 
Our game can be thought of a special case of this game with one link, \ie $|L|=1$.

The system-level problem defined\footnote{The constraint `$\bfA x \leq \bfC$' is not used in the system problem in 
\cite{alpcan02game}. This constraint is irrelevant in their case since the costs approach infinity as this constraint becomes 
active.} in~\cite{alpcan02game} however is {\em not} to maximize the sum of the objectives $\sum \varphi_i$, 
but instead to maximize the aggregate utility less the total delay over the network. Specifically, 
an {\em efficient allocation} in this case is the solution of following problem.
$$
        \maxproblem{{\rm SYS-AB}}
        {x}
        {\sum_{i\in \Nscr} U_i(x_i) - \sum_{\ell \in L} c_\ell(\sum_{j: \ell \in L_j} x_j)  }
                                 {\begin{array}{r@{\ }c@{\ }l}
                \bfA x &\leq& \bfC, \\
                x & \geq& 0. 
        \end{array}}
        $$
Notice that while the delay on a particular link $\ell$ appears in the optimization problem of possibly 
multiple players, it is accounted for only once in (SYS-AB). The motivation for this is that, while each player 
that experiences the delay on $\ell$ suffers a cost $c_\ell$ due to it, the 
{\em system} goal is to merely minimize total delay over all \textit{links} (not players), in addition to maximizing aggregate utility. In particular, 
in (SYS-AB) the delay on link $\ell$ is not scaled by the number of 
players using it. Under the assumption that $c_\ell$ approaches infinity as the link $\ell$ gets congested, it is proved 
in~\cite{alpcan02game} that the GNE of this game is efficient. The GNE and the VE coincide in this case 
because they are in the interior of the feasible region~\cite{kulkarni09refinement}. 
\end{examplec}

We derive a general principle within our setting for a result such as that in~\cite{alpcan02game} 
to hold. We stick to the solution concept of the VE and show that for utility functions 
taken from the class $\Fscr'$, the VE of the \scg is optimal for the system problem analogous 
to the one above. 

Consider the game where player $i$ derives utility  $U_i(x)$ from an 
allocation $x$ but he suffers a cost $\tau(x)$, which we assume is a convex function of $x$. 
Player $i$'s objective then is to solve 
$$
        \maxproblem{A$_i(x\mi)$}
        {x_i}
        {U_i(x) - \tau(x)}
                                 {\begin{array}{r@{\ }c@{\ }l}            
                \sum_{j \in \Nscr} x_j & \leq & C, \\
                x & \geq& 0. 
        \end{array}}
        $$
(SYS) problem  is changed to the following.
$$
        \maxproblem{{\rm SYS}}
        {x}
        {\sum_{j\in \Nscr} U_j(x) - \tau(x) }
                                 {\begin{array}{r@{\ }c@{\ }l}            
                \sum_{j \in \Nscr} x_j & \leq & C, \\
                x & \geq& 0. 
        \end{array}}
        $$
\def\bfU{{\bf U}}

Let $\bfU$ be the tuple $(U_1,\hdots,U_N)$ and let $F$ be as before,
$$ F= - \pmat{\nabla_1 U_1 -\nabla_1 \tau \\ \vdots \\ \nabla_N U_N -\nabla_N \tau }.$$
 We now ask the question: for what utility functions $\bfU$ 
does the VE of the \scg formed from (A$_1),\hdots,($A$_N$) solve \SYS?
Similar to the argument in Section~\ref{sec:best}, suppose there exists a concave continuously differentiable 
function $f$ such that $-\nabla f = F$. It follows that there exist for each $i$ 
continuously differentiable functions $\eta_i$ of $x\mi$ such that
$$f(x) = U_i(x) - \tau(x) + \eta_i(x\mi) \qquad \forall \ x.$$
For the VE to solve \SYS, we need that $f =\Theta$, where $\Theta$ now is 
$$\Theta = \sum_{j \in \Nscr} U_j -\tau.$$
$\Theta$ is a concave function if $\bfU$ satisfies Assumption~\ref{5-assump}. 
Substituting $f=\Theta$ above gives $$ \sum_{j \in \Nscr} U_j(x) =  U_i(x) + \eta(x\mi), \qquad \forall \ x,$$
for all $i$. The above equation is the same as \eqref{5-eq:2}, but for functions $\bfU$ instead of $\Phi$. 
It follows that for each $i$, $U_i$ must be of the form given by \eqref{eq:slade}, and in particular must belong to class $\Fscr'$. This is summarized 
in the following result. 

\begin{theorem}
 Suppose $\bfU$ satisfies Assumption~\ref{5-assump} and consider the game formed from 
 $({\rm A}_1),$ $\hdots,({\rm A}_N)$ described above.  The identity $-\nabla \Theta = F$ holds if and only if $\bfU$ belongs
to the class  
 $\Fscr'$. 
\end{theorem}
It is easy to see that the game of Example~\ref{ex:alpcan} follows as a special of this theorem.

\subsection{Reserve price}
So far we have assumed that there is no administrative intervention in the allocation of the resources. 
A consequence of this is the possibility that, at equilibrium, ``non-serious'' players with low marginal utility get
 large portions of the resource. An apt instance of this situation is the game in 
Example~\ref{ex:wcgne}, where at equilibrium, the player with the highest marginal utility amongst all players 
gets no portion of the resource while the player with lowest marginal utility gets all of it. The system 
problem, on the other hand, allocates all the resource to player with highest marginal utility. Since the lowest 
marginal utility could be arbitrarily small, the worst case efficiency is zero. 

One way to mitigate this problem is to {\em screen} players so that, at equilibrium,  only those players 
play the game who are ``sufficiently serious'' about wanting the resource. We show here that this can be 
done with some, though minimal administrative intervention and that under certain assumptions, efficiency as high as unity is 
achievable. Underlying this approach, is the concept 
of a {\em reserve price}. 

The idea of a reserve price or reservation price can be traced to Myerson's seminal work on optimal auction design for a single indivisible 
object in a Bayesian setting~\cite{myerson81optimal}. In this setting, 
a reserve price is set by the auctioneer or seller and revealed to all players with the rider that 
only those bids will be considered in the auction that are at least large as the reserve price. As a consequence, 
players whose valuations do not exceed the reserve price are eliminated from the auction. The seller risks 
keeping the object to himself in the event that the all players bid below this price but he also has a chance for higher revenue\footnote{See~\cite{engelbrecht-wiggans87optimal,nisan07algorithmic} for an analysis of the revenue-optimal 
reserve price in auctions.}. 

In the context of allocation of divisible resources that we are concerned with, the reserve price is typically implemented 
through the pricing formula or mechanism employed, as has been done, \eg, by Maheswaran and Ba\c{s}ar~\cite{maheswaran03nash,maheswaran06efficient}. 
In our setting of a \scg we do not have a pricing formula or a direct handle on the price. We impose the price through 
a cost or a toll. The administrator fixes a value $\pi$ so that each player $i$ that receives quantity $x_i$  
is charged a cost $\pi x_i$. This cost may be monetary and collected by an administrator or it may 
be a cost induced by a physical inconvenience such as delay which does not require imposition by an administrator. 
We do not concern ourselves with the manner of imposition of this cost, but only note that even in the case where 
an administrator imposes it, it can be imposed with minimal intervention.
As a result of the reserve price, player $i$ now solves the following optimization problem.
$$
        \maxproblem{A$_i(x\mi)$}
        {x_i}
        {\varphi_i(x) - \pi x_i}
                                 {\begin{array}{r@{\ }c@{\ }l}            
                \sum_{j \in \Nscr} x_j & \leq & C, \\
                x_i & \geq& 0. 
        \end{array}}
$$
We benchmark efficiency with respect to the original system problem.
$$
        \maxproblem{{\rm SYS}}
        {x}
        {\sum_{j\in \Nscr} \varphi_j(x)  }
                                 {\begin{array}{r@{\ }c@{\ }l}            
                \sum_{j \in \Nscr} x_j & \leq & C, \\
                x & \geq& 0. 
        \end{array}}
        $$

Consider utility functions $\Phi \in \Fscr'$. The argument in Section~\ref{sec:wcgne} shows, it suffices to consider 
utility functions $\Phi \in \Lscr'$ to evaluate the worst case efficiency, so we restrict ourselves to linear 
utility functions. In particular, consider 
for each $i$, $\varphi_i(x) = d_i^Tx$ for all $x$ and $d_i^i>0$. Notice that for any $i$, coefficients 
$d_i^j$, with $j\neq i$ do not affect optimal choice of player $i$ and the resulting equilibrium. For simplicity, we consider `perfectly competitive' utility functions, \ie, $d^j_i =0$ for $j \neq i.$

If the reserve price is larger than 
the marginal utility of a player $\hat{i}$, \ie $\pi > c_{\hat{i}}$, then it is optimal for $\hat{i}$ 
to ask for zero quantity. For the other players, the shared constraint is now the set 
$$ \left\{x^{-\hat{i}}\  \left\lvert  \ x^{-\hat{i}} \geq 0, \sum_{j \neq \hat{i}} x_j \leq C \right. \right\}. $$
Effectively, player $\hat{i}$ is eliminated and we are left with a game similar to the original game, but 
amongst players $\Nscr \backslash \hat{i}$. 
The imposition of a reserve price filters players with ``low interest'' 
in the resource and retains only those players who gain utility at least $\pi$ from 
a unit of the resource. 

There is, of course, a possibility of {\em overdoing} 
this elimination by eliminating {\em all} players in the competition. 
Indeed, for a given $c$, one can always find a price 
$\pi = \pi^*,$ where $\pi^*  > \max_i c_i$ so that for  each player, it is optimal to demand zero, 
and no resource is allocated. \ie the allocation $(0,\hdots,0)$ is the (unique) equilibrium. 
This is akin to the nonzero probability in the Bayesian single object auction, that the seller has to keep the item for itself with 
zero revenue. In that setting, one can counter possibility by arguing that the {\em expected} profit of the seller is 
positive for a certain price~\cite{nisan07algorithmic}. 
To rule out this possibility in our setting, we are compelled to assume that there is at least one player who is not eliminated. With 
this assumption, the imposition of a reserve price leads to improvement in efficiency. Indeed arbitrarily high efficiencies 
are achievable. We show this below.

Consider the GNE as a solution concept and assume that $\Phi \in \Lscr'$ is given by linear functions as above. 
In particular, let $c$ be such that $c_1=\max_i c_i =1$ and let the reserve price $\pi$ be a number in $(0,1)$. 
Let $x$ be a GNE of this game. \ie, suppose there exists $\Lambda = (\lambda_1,\hdots,\lambda_N)$ such that 
\begin{align*}
 0 \leq x &\perp -c + \pi \bfone + \Lambda \geq 0 \\ 
 0 \leq \Lambda & \perp C - \bfone^Tx \geq 0
\end{align*}
Since $\pi<1=c_1$, at equilibrium, the Lagrange multiplier $\lambda_1$ for player $1$, 
has to be positive. The complementarity between $\lambda_1$ and `$C -\bfone^Tx$' now ensures 
that $\bfone^Tx =C$ holds for the GNE $x$. Therefore at least one component of $x$ is positive. 
If $x$ is a VE, \ie a GNE with $\lambda_1 = \lambda_j$ 
for all $j$, those players $i$ with $c_i = c_1$ receive positive quantity, whereas the rest receive zero. 
If $x$ is a GNE, denote by $\Iscr$ the set of players 
with marginal utility at least $\pi$, \ie,
$$ \Iscr = \{ i \in \Nscr\ | c_i \geq \pi \}, \quad \aur~\mbox{let} \quad c_{\Iscr} = (c_i)_{i \in \Iscr}, \quad x_{\Iscr} = (x_i)_{i \in \Iscr}, $$
then $x_{\Iscr^c}$, which is the tuple of $x_j$ so that $c_j < \pi$, is the tuple of $|\Iscr^c|$ zeroes.  Consequently, 
$\sum_{j\in \Iscr} x_{j} =C$.
Therefore for this game and a GNE $x$,   
$$ \frac{\Theta(x)}{\max_{z \in \Cbb} \Theta(z)} = \frac{c^Tx}{\max_{z \in \Cbb} c^Tz} = \frac{c_{\Iscr}^Tx_{\Iscr}}{C} \geq
\frac{\pi \sum_{j \in \Iscr}x_{j}}{C}=\pi.$$ 
\ie the efficiency of the GNE is at least $\pi$. Whereas every VE $x$ is efficient:
$$ \frac{\Theta(x)}{\max_{z \in \Cbb} \Theta(z)} = \frac{c^Tx}{\max_{z \in \Cbb} c^Tz} = \frac{c_1\bfone^Tx}{C} =1.$$
This is true for for all $c >0$ such that $\max_i c_i =1$ 
when the reserve price $\pi <1$. More generally, for any $c>0$ and a reserve price $\pi$, we have 
$$ \inf_{x \in \GNE(c)} \frac{\Theta(x)}{\max_{z \in \Cbb} \Theta(z)} \ \  \begin{cases}
                                                                 \geq \frac{\pi}{\max_i c_i} \quad &        \pi < \max_i c_i, \\
                                                                  = 0 \quad & \pi \geq \max_i c_i.
                                                                        \end{cases}
$$
When $\max_i c_i =c_1 = 1$ and $\pi \in (0,1)$, note that $\pi$ is a lower bound and the actual efficiency may in fact be greater than $\pi$. 
For \eg, consider a game where $\max_{j \neq 1} c_j < 1$ and $1>\pi > \max_{j \neq 1} c_j$. 
For this game, all players other than player $1$ are ``eliminated'', \ie they receive zero quantity at equilibrium. 
Therefore the efficiency of the GNE in this case is unity.

\section{Conclusions} \label{sec:final}
This work considered resource allocation through a \scg from the point of view of 
economic efficiency. This game is relevant 
in the setting where an auctioneer is not available for operationalizing a mechanism. 
We clarified the relationship of this game with other modes of allocating resources. 
\Scgs admit two kinds of equilibria, namely, the GNE and the VE. We 
considered a class of concave objective functions and found that the worst case 
efficiency over this class of both, the GNE and the VE, is zero. However we show that 
there is subclass for which the VE is always efficient but the GNE can be arbitrarily inefficient. This further corroborates the thesis put forth in~\cite{kulkarni09refinement} that the VE should be considered a refinement of the GNE. 
We then discussed remedies by which the worst case efficiency may be bounded away from zero.
Specifically, we showed that  utility functions with bounded gradients, alternative notions of 
efficiency and the imposition of a reserve price can mitigate the loss of efficiency.

\appendix
\subsection{Proof of Lemma~\ref{lem:linear}}
\begin{proof}
Since $\Theta$ is concave, the following inequality holds:
\begin{equation} 
\Theta(\bar{x}) \leq \Theta(x) + \nabla \Theta(x)^T(\bar{x}-x). \label{eq:concave}
\end{equation} Now consider the ratio $\frac{\Theta(x)}{\Theta(x^*)}$. 
Adding and subtracting $\nabla \Theta(x)^Tx$ in the numerator and using \eqref{eq:concave} with $\bar{x} = x^*$,  it follows
that
this ratio satisfies  
\begin{equation}
\frac{\Theta(x)}{\Theta(x^*)} \geq \frac{[\Theta(x) - \nabla \Theta(x)^Tx] + \nabla\Theta(x)^Tx}{[\Theta(x) -
\nabla\Theta(x)^Tx] + \nabla\Theta(x)^Tx^*}.  \label{eq:concave1}
\end{equation}
By definition of $x^\ell$, $\nabla\Theta(x)^Tx^\ell \geq \nabla \Theta(x)^Tx^*,$
and since both these terms are positive, using this inequality in \eqref{eq:concave1} gives
\begin{equation}
\frac{\Theta(x)}{\Theta(x^*)} \geq \frac{[\Theta(x) - \nabla \Theta(x)^Tx] + \nabla\Theta(x)^Tx}{[\Theta(x) -
\nabla\Theta(x)^Tx] + \nabla\Theta(x)^Tx^\ell}\label{eq:concave2} \end{equation}
Now, from \eqref{eq:concave}, taking $\bar{x} =0$, we get
\beq
\Theta(x) -\nabla \Theta(x)^Tx \geq \Theta(0), \label{eq:concave0}
\eeq which is nonnegative, by Assumption~\ref{5-assump}. Furthermore Assumption~\ref{5-assump} also provides that 
$\nabla \Theta(x)$ is nonnegative whereby $\nabla \Theta(x)^Tx$ and $\nabla\Theta(x)^Tx^\ell$ are both nonnegative. 
So therefore, dropping the nonnegative term `$\Theta(x) - \nabla\Theta(x)^Tx$' from the numerator  
and denominator of \eqref{eq:concave2} and recalling that $\nabla \Theta(x)^Tx^\ell \geq \nabla \Theta(x)^Tx$, 
we obtain,
$$ \frac{\Theta(x)}{\Theta(x^*)} \geq \frac{\nabla \Theta(x)^Tx}{\nabla\Theta(x)^Tx^\ell}, $$
which is the desired result.
\end{proof}
\subsection{Proof of Theorem~\ref{thm:lscrprime}}
\begin{proof}
Lemma~\ref{lem:linear} shows that for  $\Phi \in \Fscr$ and $\bar{x} \in \VE(\Phi)$,
\begin{align*} 
\frac{\Theta(\bar{x})}{\max_{z \in \Cbb}\Theta(z)} &\geq \frac{\widetilde{\Theta}^{\bar{x}}(\bar{x})}{\max_{z \in
\Cbb}\widetilde{\Theta}^{\bar{x}}(z)}. 
\end{align*}
So for any $\Phi \in \Fscr$,\begin{align*} 
\inf_{\bar{x} \in \VE(\Phi)} \frac{\Theta(\bar{x})}{\max_{z \in \Cbb}\Theta(z)} &\geq \inf_{\bar{x} \in \VE(\Phi)} 
\frac{\widetilde{\Theta}^{\bar{x}}(\bar{x})}{\max_{z \in \Cbb}\widetilde{\Theta}^{\bar{x}}(z)} \\
&\geq \inf_{\bar{x}\in \VE(\widetilde{\Phi}^{x}),\: x\in \VE(\Phi)}  \frac{\widetilde{\Theta}^{\bar{x}}(\bar{x})}{\max_{z \in
\Cbb}\widetilde{\Theta}^{\bar{x}}(z)}, \quad \mbox{by Lemma~\ref{lem:subset}},\\ 
&= \inf_{\bar{x}\in \VE(\widetilde{\Phi}^{x}),\: x\in \VE(\Phi)}  \frac{\sum_{i \in \Nscr}\widetilde{\varphi}_i^{\bar{x}}(\bar{x})}{\max_{z \in \Cbb}\sum_{i \in \Nscr} \widetilde{\varphi_i}^{\bar{x}}(z)}, \\
&\geq  \min_{\bar{x} \in \VE(\Phi),\: \Phi \in \Lscr} \frac{\Theta(\bar{x})}{\max_{z \in \Cbb}\Theta(z)}, \quad \mbox{since }
\{ \widetilde{\Phi}^x : x \in \VE(\Phi) \} \subseteq \Lscr.
\end{align*}
Since this holds for any $\Phi \in \Fscr$, we must have
$$\inf_{\bar{x} \in \VE(\Phi), \: \Phi \in \Fscr} \frac{\Theta(\bar{x})}{\max_{z \in \Cbb}\Theta(z)} \geq \inf_{\bar{x} \in
\VE(\Phi),\: \Phi \in \Lscr} \frac{\Theta(\bar{x})}{\max_{z \in \Cbb}\Theta(z)}. $$
But since $\Lscr \subseteq \Fscr$, we must also have
$$\inf_{\bar{x} \in \VE(\Phi), \: \Phi \in \Fscr} \frac{\Theta(\bar{x})}{\max_{z \in \Cbb}\Theta(z)} \leq \inf_{\bar{x} \in
\VE(\Phi),\: \Phi \in \Lscr} \frac{\Theta(\bar{x})}{\max_{z \in \Cbb}\Theta(z)}. $$
The result follows.
\end{proof}

\subsection{Proof of Theorem~\ref{thm:lscrprimeprime} }
\begin{proof}
The proof is similar to that of Theorem~\ref{thm:lscrprime}, so we will only sketch it. 
By repeating the arguments in Lemma~\ref{lem:subset}, one can conclude that for any $\Phi \in \Fscr'$,
 \begin{equation}
\GNE(\Phi) \subseteq \bigcup_{x \in \GNE(\Phi)} \GNE(\widetilde{\Phi}^{x}), \label{eq:vesubset}
 \end{equation}
where the notation $\widetilde{\Phi}^{x}$ stands for the linearized version of $\Phi$, as is Section~\ref{sec:wcve}. 
By Lemma~\ref{lem:linear}, for any $\Phi \in \Fscr'$ and $\bar{x} \in \GNE(\Phi)$,
\begin{align*} 
\frac{\Theta(\bar{x})}{\max_{z \in \Cbb}\Theta(z)} &\geq \frac{\widetilde{\Theta}^{\bar{x}}(\bar{x})}{\max_{z \in
\Cbb}\widetilde{\Theta}^{\bar{x}}(z)}. 
\end{align*}
We see that, if $\Phi$ belongs to the class $\Fscr'$, \ie if $F(x) = -\nabla \Theta(x)$ for all $x$,  then the 
linearization $\widetilde{\Phi}^{\bar{x}}$ belongs to $\Lscr'$. \ie we have $\widetilde{F}^{\bar{x}}(x) = -\nabla
\widetilde{\Theta}^{\bar{x}}(x)$ for all $x$.
Similar to the proof of Theorem~\ref{thm:lscrprime}, using \eqref{eq:vesubset}, we get that for each $\Phi \in \Fscr,$
$$\inf_{\bar{x} \in \GNE(\Phi)} \frac{\Theta(\bar{x})}{\max_{z \in \Cbb}\Theta(z)} \geq \inf_{\bar{x} \in \GNE(\Phi),\: \Phi
\in \Lscr'} \frac{\Theta(\bar{x})}{\max_{z \in \Cbb}\Theta(z)}. $$
Then using that $\Lscr' \subseteq \Fscr$ completes the proof.
\end{proof}
\bibliographystyle{plainini}
\bibliography{ref}
\end{document}